\documentclass[10pt, conference, letterpaper]{IEEEtran}
\usepackage{cite}
\ifCLASSINFOpdf
\usepackage[pdftex]{graphicx}
\else
\usepackage[dvips]{graphicx}
\fi
\usepackage[cmex10]{amsmath}
\usepackage{algorithm}
\usepackage{algorithmic}
\usepackage{amsmath}
\usepackage{array}
\usepackage{framed}
\usepackage[tight,footnotesize]{subfigure}
\usepackage{multirow}
\usepackage{fancyhdr}
\usepackage{tikz}
\usetikzlibrary{calc}
\usepackage{amsmath,amsthm,amssymb,amsfonts}
\usepackage[colorlinks, bookmarks, linkcolor=red, citecolor=blue]{hyperref}
\usepackage[paperwidth=8.5in, paperheight=11in, margin=0.8in]{geometry}

\newtheorem{theorem}{Theorem}
\newtheorem{lemma}{Lemma}
\newtheorem{definition}{Definition}

\newtheorem{corollary}{Corollary}

\let\emptyset\varnothing
\IEEEoverridecommandlockouts
\begin{document}

\title{Network Utility Maximization in Adversarial Environments}
\author{
\IEEEauthorblockN{Qingkai Liang and Eytan Modiano}
\IEEEauthorblockA{Laboratory for Information and Decision Systems\\Massachusetts Institute of Technology, Cambridge, MA}
\thanks{This work was supported by NSF Grant CNS-1524317 and by DARPA I2O and Raytheon BBN Technologies under Contract No. HROO l l-l 5-C-0097.}
}
\maketitle

\begin{tikzpicture}[remember picture, overlay]
\node at ($(current page.north) + (-3in,-0.5in)$) {Technical Report};
\end{tikzpicture}

\begin{abstract}
Stochastic models have been dominant in network optimization theory for over two decades, due to their analytical tractability. However, these models fail to capture non-stationary or even adversarial network dynamics which are of increasing importance for modeling the behavior of networks under malicious attacks or characterizing short-term transient behavior. In this paper, we consider the network utility maximization problem in adversarial network settings. In particular, we focus on the tradeoffs between total queue length and utility regret which measures the  difference in network utility between a causal policy and an ``oracle" that knows the future within a finite time horizon. Two adversarial network models are developed to characterize the adversary's behavior. We provide lower bounds on the tradeoff between utility regret and queue length under these adversarial models, and analyze the performance of two control policies (i.e., the Drift-plus-Penalty algorithm and the Tracking Algorithm).
\end{abstract}

\section{Introduction}
Stochastic network models have been dominant in network optimization theory for over two decades, due to their analytical tractability. For example, it is often assumed in wireless networks that the variation of traffic patterns and the evolution of channel capacity follow some stationary stochastic process, such as the i.i.d. model and the ergodic Markov model. Many important network control policies (e.g., MaxWeight \cite{tassiulas} and Drift-plus-Penalty policy \cite{neely2008fairness}) have been derived to optimize network performance  under those stochastic network dynamics.

However, non-stationary or even adversarial dynamics have been of increasing importance in recent years. For example,  modern communication networks frequently suffer from Distributed Denial-of-Service (DDoS) attacks or jamming attacks \cite{security-survey}, where traffic injections and channel conditions are controlled by some malicious entity in order to degrade network performance.
As a result, it is important to develop efficient control policies that optimize network performance even in adversarial settings. However, extending the traditional stochastic network optimization framework to the adversarial setting is non-trivial because many important notions and analytical tools developed for stochastic networks cannot be applied in adversarial settings. For example, traditional stochastic network optimization focuses on long-term network performance while in an adversarial environment the network may not have any steady state or well-defined long-term time averages. Thus, typical steady-state analysis and many equilibrium-based notions such as the network throughput region  cannot be used in networks with adversarial dynamics, and it is important to understand ``transient" network performance within a finite time horizon in a non-stationary/adversarial environment.

In this paper, we investigate efficient network control policies that can maximize network utility within a finite time horizon while keeping the total queue length small in an adversarial environment. In particular, we focus on the following optimization problem:
\begin{equation}\label{eq:intro}
\begin{split}
\max_{\alpha_t\in\mathcal{D}_{\omega_t}} \quad &\sum_{t=0}^{T-1} U(\alpha_t,\omega_t)\\
\text{s.t.}\quad & \sum_{t=0}^{T-1}a_i(t)\le \sum_{t=0}^{T-1}\tilde{b}_i(t),~\forall i
\end{split}
\end{equation}
where $U(\alpha_t,\omega_t)$ is the network utility gained in slot $t$ under the control action $\alpha_t$ (constrained to some action space $\mathcal{D}_{\omega_t}$) and the network event $\omega_t$ (which includes information about exogenous arrivals, link capacities, etc). The sequence of network events $\{\omega_t\}_{t=0}^{T-1}$ follows an arbitrary (possibly adversarial) process. The objective is to maximize the total network utility gained within a finite time horizon $T$ subject to the constraint that for each queue $i$ the total arrivals $\sum_{t=0}^{T-1}a_i(t)$ do not exceed the total departures $\sum_{t=0}^{T-1}\tilde{b}_i(t)$  during the time horizon. 
\subsection{Main Results}
We develop general adversarial network models and propose a new finite-time performance metric, referred to as \emph{utility regret} (the formal definition is given in Section \ref{sec:metric}):
\[
\mathcal{R}^\pi_T = \sum_{t=0}^{T-1} U(\alpha^*_t,\omega_t)- \sum_{t=0}^{T-1} U(\alpha^\pi_t,\omega_t),
\]
where $\{\alpha^\pi_t\}_{t=0}^{T-1}$ is the sequence of control actions taken by a policy $\pi$, and $\{\alpha^*_t\}_{t=0}^{T-1}$ is the optimal sequence of actions for solving \eqref{eq:intro} generated by an ``oracle" that knows the future. Note that a control policy $\pi$ may trivially maximize the network utility by simply ignoring the constraint in \eqref{eq:intro} (e.g., admitting all the exogenous traffic) such that the utility regret become zero or even negative\footnote{The negative utility regret may occur since any optimal solution $\{\alpha^*_t\}_{t=0}^{T-1}$ is required to satisfy the constraint in \eqref{eq:intro} while an arbitrary policy $\pi$ may violate this constraint.}. However, such an action may significantly violate the constraint in \eqref{eq:intro} and lead to large queue length. Therefore, there is a tradeoff between the utility regret and the queue length achieved by a control  policy, which is similar to the well-known utility-delay tradeoff in traditional stochastic network optimization \cite{neely-SNO}. In this paper, we investigate this tradeoff in an adversarial environment. The main results are as follows.
\vspace{1mm}

\noindent $\bullet$ We prove that it is impossible to simultaneously achieve both ``low" utility regret and ``low" queue length if the adversary is unconstrained. In particular, there exist some adversarial network dynamics such that either the utility regret or the total queue length grows at least linearly with the time horizon $T$ under any causal control policy. This impossibility result motivates us to study constrained adversarial dynamics.

\vspace{1mm}

\noindent $\bullet$  We develop two adversarial network models where the network dynamics are constrained to some ``admissible" set. In particular, we first consider the $W$-constrained adversary model, where under the optimal  policy the total arrivals do not exceed the total services within any window of $W$ slots. We then propose a more general adversary model called $V_T$-constrained adversary, where the total queue length generated by the ``oracle" during its sample path is upper bounded by $V_T$. By varying the values of $V_T$, the proposed $V_T$-constrained model covers a wide range of adversarial settings: from a strictly constrained adversary to a fully unconstrained adversary. 

\vspace{1mm}

\noindent $\bullet$  We develop lower bounds on the tradeoffs between utility regret and queue length under both the $W$-contrained and the $V_T$-constrained adversary models.  It is shown that no causal policy can simultanesouly achieve both sublinear utility regret and sublinear queue length if $W$ or $V_T$ grows linearly with $T$. We also analyze the tradeoffs achieved by two control algorithms: the Drift-plus-Penalty algorithm \cite{neely2008fairness} and the Tracking Algorithm  \cite{andrews-zhang-1, andrews-zhang-2} under the two adversarial models. In particular, both algorithms simultaneously achieve sublinear utility regret and sublinear queue length whenever $W$ or  $V_T$ grows sublinearly with $T$, yet the theoretical regret bound under the Tracking Algorithm is better than that under the Drift-plus-Penalty algorithm. The Tracking Algorithm also asymptotically attains the optimal tradeoffs under the $W$-constrained adversary model. 

\subsection{Related Work}
The study of adversarial network models dates back more than two decades ago. Rene Cruz \cite{PSM} provided the first concrete example of networks with adversarial dynamics, which were later generalized by Borodin \emph{et al.} \cite{AQT} under the \emph{Adversarial Queuing Theory} (AQT) framework. In AQT, in each time slot, the adversary injects a set of packets at some of the nodes. In order to avoid trivially overloading the system, the AQT framework imposes a stringent \emph{window constraints}: the maximum traffic injected in every link over any window of $W$ time slots should not exceed the amount of traffic that the link can serve during that interval. Andrews \emph{et al.} \cite{leaky-bucket} introduced a more generalized adversary model known as the \emph{Leaky Bucket} (LB) model that differs from AQT by allowing some traffic burst during any time interval. The AQT model and the LB model have given
rise to a large number of results since their introduction, most of which are about network stability under several simple scheduling policies such as FIFO (see \cite{AQT-survey} for a  review of these results).

However, the AQT and the LB models assume that only packet injections are adversarial while the underlying network topology and link states remain fixed. Such a static network model does not capture many adversarial environments, such as wireless networks under jamming attacks where the adversary can control the channel states. Andrews and Zhang \cite{andrews-zhang-1, andrews-zhang-2} extended the AQT model to single-hop dynamic wireless networks, where both packets injections and link states are controlled by an adversary, and prove the stability of the MaxWeight algorithm in this context. Jung \emph{et al.} \cite{max-weight-1, max-weight-2} further extended the results of \cite{andrews-zhang-1, andrews-zhang-2} to multi-hop dynamic networks. Our window-based $W$-constrained model is inspired by and similar to the adversarial models used in \cite{andrews-zhang-1, andrews-zhang-2, max-weight-1, max-weight-2}.


While the above-mentioned works focused on network stability, our work is most related to the universal network utility maximization problem by Neely \cite{neely-universal} where network utility needs to be maximized  subject to stability constraints under adversarial network dynamics. Algorithm (time-average) performance is measured with respect to a so-called ``$W$-slot look-ahead policy". Such a policy has perfect knowledge about network dynamics over the next $W$ slots but it is required that under this policy the total arrivals to each queue should not exceed the total amount of service offered to that queue during every window of $W$ slots. As a result, it is similar to our $W$-constrained model where stringent window constraints have to be enforced. 

Our paper expands previous work in a number of fundamental ways.
First, we develop lower bounds on the tradeoffs between utility regret and queue length under both the $W$-contrained and the $V_T$-constrained adversary models. As far as we know, none of the existing works (e.g., \cite{max-weight-1, max-weight-2, andrews-zhang-1, andrews-zhang-2,neely-universal}) provide lower bounds in any kind of adversarial network models.
Second, we provide analysis under the new $V_T$-constrained adversary model which generalizes the adversarial network dynamics models used by existing works.  To the best of our knowledge, existing works (e.g., \cite{max-weight-1, max-weight-2, andrews-zhang-1, andrews-zhang-2,neely-universal}) all use the $W$-constrained adversary model or similar windows-based variants due to its analytical tractability. In this paper, we propose a new $V_T$-constrained adversary model which gets rid of the window constrains. Due to the lack of window-based structure, the analysis carried out in existing works cannot be applied to the $V_T$-constrained model. We develop new analytical results under the new $V_T$-constrained model by converting the $V_T$-constrained model to a $W$-constrained model with a carefully selected window size $W$.
\subsection{Organization of this Paper}
The rest of this paper is organized as follows. We first introduce the system model and relevant performance metrics in Section \ref{sec:model}. 	We study the $W$-constrained and $V_T$-constrained adversary models in Sections \ref{sec:W} and \ref{sec:VT}, respectively. Finally, simulation results and conclusions are given in Sections \ref{sec:sim} and \ref{sec:conclusion}, respectively.
\section{System Model}\label{sec:model}
Consider a network with $N$ queues (the set of all queues are denoted by $\mathcal{N}=\{1,\cdots,N\}$). Time is slotted with a finite horizon $\mathcal{T}=\{0,\cdots,T-1\}$. Let $\omega_t$  denote the \emph{network event} that occurs in  slot $t$, which indicates the current network parameters, such as a vector of conditions for each link, a vector of exogenous arrivals to each node, or other relevant information about the current network links and exogenous arrivals. The set of all possible network events is denoted by $\Omega$.

At the beginning of each time slot $t$, the network operator observes the current network event $\omega_t$ and chooses a control action $\alpha_t$ from some action space $\mathcal{D}_{\omega_t}$ that can depend on $\omega_t$. The network event $\omega_t$ and the control action $\alpha_t$ together produce the service vector $\mathbf{b}(\alpha_t, \omega_t)\triangleq \mathbf{b}(t)=(b_1(t),\cdots,b_N(t))$ and the arrival vector $\mathbf{a}(\alpha_t,\omega_t)\triangleq \mathbf{a}(t)=(a_1(t),\cdots,a_N(t))$. Note that $a_i(t)$ includes both the admitted exogenous arrivals from outside the network to queue $i$, and the endogenous arrivals from other queues (i.e., routed packets from other queues to queue $i$). Thus, the above network model accounts for both single-hop and multi-hop networks, and the control action $\alpha_t$ may correspond to, for example, joint admission control, routing, rate allocation and scheduling decisions in a multi-hop network.
Let $\mathbf{Q}(t)=(Q_1(t),\cdots,Q_N(t))$ be the queue length vector at the beginning of slot $t$ (before the arrivals in that slot). The queueing dynamics are 
\[
Q_i(t+1)=[Q_i(t)+a_i(t)-b_i(t)]^+,~\forall i\in\mathcal{N},t\in\mathcal{T},
\]
where $[x]^+=\max\{x,0\}.$

We assume that the sequence of network events $\{\omega_t\}_{t=0}^{T-1}$ are generate according to an \emph{arbitrary} process (possibly non-stationary or even adversarial), except for the following boundedness assumption. Under any network event and any control action, the arrivals and the service rates in each slot are bounded by  constants that are independent of the time horizon $T$: for any $\omega_t\in\Omega$ and any $\alpha_t\in\mathcal{D}_{\omega_t}$
\[
0\le a_i(\alpha_t,\omega_t)\le A,~~0\le b_i(\alpha_t,\omega_t)\le B.
\]
For simplicity, we assume $B\ge A$ such that both arrivals and services are upper bounded by $B$ in each slot.

A policy $\pi$ generates a sequence of control actions $\big(\alpha^\pi_0,\cdots,\alpha^\pi_{T-1}\big)$ within the time horizon. In each slot $t$, the queue length vector, the arrival vector and the service rate vector under policy $\pi$ is denoted by $\mathbf{Q}^\pi(t)$, $\mathbf{a}^\pi(t)$ and $\mathbf{b}^\pi(t)$, respectively. A \emph{causal} policy is one that generates the current control action $\alpha_t$ only based on the knowledge up until the current slot $t$. In contrast, a \emph{non-causal} policy may generate the current control action $\alpha_t$ based on knowledge of the future.

Let $U(\alpha_t,\omega_t)$ be the network utility gained in slot $t$ if action $\alpha_t$ is taken under network event $\omega_t$. We assume that under any control action and any network event, network utility is bounded:
\[
U_{\min}\le U(\alpha_t,\omega_t)\le U_{\max},~\forall \omega_t\in\Omega, \alpha_t\in\mathcal{D}_{\omega_t}.
\]
A commonly-used form of the network utility function is $U(\alpha_t,\omega_t)=\sum_i U_i\big(x_i(t)\big)$ where $x_i(t)$ is the amount of admitted exogenous traffic to queue $i$ in slot $t$. Typical examples  include $U(\alpha_t,\omega_t)=\sum_i x_i(t)$ (total throughput), $U(\alpha_t,\omega_t)=\sum_i \log\big(x_i(t)\big)$ (proportional fairness), etc. In wireless networks with power control, another widely-used network utility function is $U(\alpha_t,\omega_t)=-\sum_i P_i(t)$ where $P_i(t)$ is the power allocated to queue $i$ in slot $t$. This utility function aims to minimize the total power consumption.

In this paper, we consider the following network utility maximization problem, referred to as \textbf{NUM}.
\vspace{-2mm}
\begin{framed}
\noindent \textbf{NUM:}\vspace{-3mm}
\small
\begin{align}
\max_{\alpha_t\in\mathcal{D}_{\omega_t}} \quad &\sum_{t=0}^{T-1} U(\alpha_t,\omega_t)\label{eq:ano1}\\
\text{s.t.}\quad & \sum_{t=0}^{T-1}a_i(t)\le \sum_{t=0}^{T-1}\tilde{b}_i(t),~\forall i\in\mathcal{N}\label{eq:ano2},
\end{align}\vspace{-4mm}
\end{framed}
\noindent where $\tilde{b}_i(t)=\min\{b_i(\alpha_t,\omega_t),Q_i(t)\}$ is the actual packet departures from queue $i$ in slot $t$. The objective \eqref{eq:ano1} is to maximize the total network utility gained in the time horizon. The constraint \eqref{eq:ano2} requires that the total arrivals to each queue  should not exceed the total amount of 
departures from that queue during the time horizon. 
Note that the above optimization problem is a natural analogue of the traditional stochastic network optimization problem \cite{neely-SNO}, where the time-average utility is maximized subject to certain network stability constraints. Indeed, if we consider a stochastic network with an infinite time horizon, then the objective \eqref{eq:ano1} is equivalent to maximizing time-average network utility, and the constraint \eqref{eq:ano2} requires that the time-average arrival rate to each queue should not exceed the time-average service rate, which is equivalent to rate stability\footnote{A network is rate-stable under a control policy $\pi$ if $\sum_{i} Q^\pi_i(T)\slash T\rightarrow0$ as $T\rightarrow\infty$.}. 

\subsection{Asymptotic Notations}
Let $f$ and $g$ be two functions defined on some subset of real numbers. Then $f(x)=O(g(x))$ if $\limsup_{x\rightarrow\infty} \frac{|f(x)|}{g(x)}<\infty$. Similarly,  $f(x)=\Omega(g(x))$ if $\liminf_{x\rightarrow\infty} \frac{f(x)}{g(x)}>0$. Also, $f(x)=\Theta(g(x))$ if $f(x)=O(g(x))$ and $f(x)=\Omega(g(x))$. In addition, $f(x)=o(g(x))$ if $\lim_{x\rightarrow\infty} \frac{f(x)}{g(x)}=0$, and in this case we say that $f(x)$ is sublinear in $g(x)$.
\subsection{Performance Metrics}\label{sec:metric}
Our objective is to find a causal control policy that can maximize the network utility while keeping the total queue length small. Note that a network with adversarial dynamics may not have any steady state or well-defined time averages. Hence, it is crucial to understand the transient behavior of the network, and the traditional equilibrium-based performance metrics may not be appropriate in an adversarial setting. As a result, we introduce the notion of \emph{utility regret} to measure the finite-time performance achieved by a control policy.
\begin{definition}[Utility Regret]
Given the time horizon $T$, the utility regret achieved by a policy $\pi$ under a sequence of network events $\omega_0,\cdots,\omega_{T-1}$ is defined to be
\begin{equation}\label{eq:R}
\small
\mathcal{R}^\pi_T\Big(\{\omega_0,\cdots,\omega_{T-1}\}\Big)=\sum_{t=0}^{T-1} U(\alpha^*_t,\omega_t)-\sum_{t=0}^{T-1} U(\alpha^\pi_t,\omega_t),
\end{equation}
where $\{\alpha^*_t\}_{t=0}^{T-1}$ is an optimal solution to \textbf{\emph{NUM}} generated by an ``oracle" that knows the entire sequence of network events $\{\omega_0,\cdots,\omega_{T-1}\}$ in advance. 
\end{definition}
\noindent In this setup, a policy $\pi$ is chosen and then the adversary selects the
sequence of network events $\{\omega_0,\cdots,\omega_{T-1}\}$ that maximize the regret. Intuitively, the notion of utility regret captures the worst-case utility difference between a causal policy and an ideal $T$-slot lookahead non-causal policy. 


Note that any optimal solution $\{\alpha^*_t\}_{t=0}^{T-1}$ to \textbf{NUM} is a utility maximizing policy subject to the constraint \eqref{eq:ano2} that it clears all the backlogs within the time horizon. 
A causal control policy may trivially  maximize the network utility  by simply ignoring the stability constraint \eqref{eq:ano2} (e.g., admitting all the exogenous traffic) such that the utility regret become zero or even negative. However, such an action may significantly violate the stability constraint \eqref{eq:ano2} and lead to large total queue length. As a result, there is a tradeoff between the utility regret and the total queue length achieved by a causal control policy. 

A desirable first order characteristic of a ``good" policy $\pi$ is that it simultaneously  achieves sublinear utility regret and sublinear queue length w.r.t. the time horizon $T$, i.e., $\mathcal{R}^\pi_T=o(T)$ and $\sum_{i} Q^\pi_i(T)=o(T)$. Sublinear utility regret guarantees  that $\mathcal{R}^\pi_T\slash T\rightarrow 0$ as the time horizon $T\rightarrow\infty$, meaning that the time-average utility gained under policy $\pi$ asymptotically approaches that under the optimal non-causal policy. In other words, the long-term time-average utility is optimal. Sublinear queue length ensures  $\sum_{i} Q^\pi_i(T)\slash T\rightarrow 0$ as $T\rightarrow\infty$, which is equivalent to rate stability. Note that simultaneously achieving sublinear utility regret and sublinear queue length is equivalent to maximizing long-term time-average utility subject rate stability, which is  the goal of traditional stochastic network optimization \cite{neely-SNO}.

Note that simultaneously achieving sublinear utility regret and sublinear queue length is just a coarse-grained requirement for a ``good" tradeoff between utility regret and queue length. In an adversarial setting with no steady state, the fine-grained growth rates of utility regret and queue length are equally important and should also be well balanced. A better tradeoff in terms of their growth rates implies that the policy has a better learning ability and can adapt to the adversarial environment faster.

Unfortunately, the following theorem shows that in general no causal policy can simultaneously achieve both sublinear utility regret and sublinear queue length.
\begin{theorem}\label{thm:impossible}
For any causal policy $\pi$, there exists a sequence of network events $\omega_0,\cdots,\omega_{T-1}$ such that either the utility regret $\mathcal{R}^\pi_T\Big(\{\omega_0,\cdots,\omega_{T-1}\}\Big)=\Omega(T)$ or the total queue length  $\sum_i Q^\pi_i(T)=\Omega(T)$.
\end{theorem}
\begin{proof}
We prove this theorem by considering a specific one-hop network with 2 users and constructing a sequence of adversarial network dynamics such that either the utility regret or the total queue length grows at least linearly with the time horizon $T$. More specifically, the time horizon is split into two parts. In the first $T\slash 2$ slots, the adversary just generates some regular network events, let the policy run and observes the queue lengths of the two users. In the remaining $T\slash 2$ slots, the adversary sets the capacity to zero for the user with a longer queue and creates sufficient capacity for the other user such that the performance of the causal policy is significantly degraded while the ``oracle" can still perform very well. See Appendix \ref{ap:impossible} for details.
\end{proof}

\vspace{1mm}

Theorem \ref{thm:impossible} shows that is it impossible to achieve sublinear utility regret while maintaining sublinear queue length, if the adversary has unconstrained power in determining the network dynamics. As a result, in the following two sections, we develop two adversary models where the sequence of network events (i.e. network dynamics) that the adversary can select is constrained to some ``admissible" set.  In Section \ref{sec:W}, we consider the $W$-constrained adversary model that is an extension of the widely-known yet very stringent model used in Adversarial Queueing Theory. Next in Section \ref{sec:VT}, we develop a more relaxed adversary model called the $V_T$-constrained adversary. Lower bounds on the tradeoffs between utility regret and queue length as well as the performance of some commonly-used algorithms are analyzed under the two adversary models.
\section{$W$-Constrained Adversary Model}\label{sec:W}
In this section, we investigate the $W$-constrained adversary model which is an extension of the classical Adversarial Queueing Theory (AQT) model \cite{AQT}. It has  stringent constraints on the set of admissible network dynamics that the adversary can set, yet is analytically tractable, which facilitates our subsequent investigation of a more relaxed adversary model in Section \ref{sec:VT}. We first give the definition of $W$-constrained network dynamics.
\begin{definition}[$W$-Constrained Dynamics]
Given a window size $W\in[1,T]$, a sequence of network events $\omega_0,\cdots,\omega_{T-1}$ is $W$-constrained if
\begin{equation}\label{eq:AQT}
\sum_{\tau=t}^{t+W-1} a^*_i(\tau)\le \sum_{\tau=t}^{t+W-1} b_i^*(\tau),~\forall i\in \mathcal{N},t\in\mathcal{T},
\end{equation}
where $\big\{\mathbf{a}^*(t)\big\}_{t=0}^{T-1}$ and $\big\{\mathbf{b}^*(t)\big\}_{t=0}^{T-1}$ is the optimal solution to \textbf{\emph{NUM}} under the above sequence of network events.
\end{definition}
\noindent Note that if there exist multiple optimal solutions to \textbf{NUM}, then constraint \eqref{eq:AQT} is only required to be satisfied by any one of them. Any network satisfying the above is called a \textbf{$W$-constrained network}. In other words, under the optimal (possibly non-causal) policy, the total amount of arrivals to each queue does not exceed  the total amount of service offered to that queue during any window of $W$ slots. 

Denote by $\mathcal{W}_T$ the set of all sequences of network events $\{\omega_0,\cdots,\omega_{T-1}\}$ that are $W$-constrained. Then the \textbf{$W$-constrained adversary} can only select the sequence of network events from the constrained set $\mathcal{W}_T$. 

In the following, we first provide a lower bound on the tradeoffs between utility regret and  queue length under the $W$-constrained adversary model (Section \ref{sec:lower-W}), and then analyze the tradeoffs achieved by several common control policies (Section \ref{sec:alg-W}). Note that throughout this section we mainly focus on the dependence of utility regret and queue length on $W$ and $T$ while \textbf{treating the number of users $N$ a constant}.
\subsection{Lower Bound on the Tradeoffs}\label{sec:lower-W}
The following theorem provides a lower bound on the tradeoffs between utility regret and  queue length under the $W$-constrained adversary model.
\begin{theorem}\label{thm:lower-AQT}
For any causal policy $\pi$, there exists a sequence of network events $\{\omega_0,\cdots,\omega_{T-1}\}\in\mathcal{W}_T$ such that 
\[
\mathcal{R}^\pi_T\Big(\{\omega_0,\cdots,\omega_{T-1}\}\Big)+c\sum_i Q_i^\pi(T)\ge c'W,
\]
where $c'>0$ is some constant.
\end{theorem}
\begin{proof}
We prove this theorem by constructing a sequence of network events such that lower bound is attained. The construction is similar to the one used in the proof of Theorem \ref{thm:impossible}. The difference is that the constructed sequence of network events is $W$-constrained here. See Appendix \ref{ap:lower-AQT} for the detailed proof.
\end{proof}

\vspace{1mm}

Note that if the window size $W$ is comparable with the time horizon $T$, i.e., $W=\Theta(T)$, the above theorem shows that no causal policy can simultaneously achieve sublinear utility regret and sublinear queue length under the $W$-constrained adversary model. On the other hand, if $W=o(T)$, there \emph{might} exist some causal policy that attains sublinear utility and subinear queue length simultaneously, which we investigate in the next section. In particular, we show that the above lower bound can be asymptotically attained by some causal policy.
\subsection{Algorithm Performance in $W$-Constrained Networks}\label{sec:alg-W}
In this section, we analyze the tradeoffs between utility regret and queue length achieved by two network control algorithms under the $W$-constrained adversary model. The first is the famous Drift-plus-Penalty algorithm \cite{neely2008fairness} that was proved to achieve good utility-delay tradeoffs in stochastic networks. The second is a generalized version of the Tracking Algorithm  \cite{andrews-zhang-1, andrews-zhang-2} that was originally proposed for Adversarial Queueing Theory. In particular, we show that the Tracking Algorithm attains the tradeoff lower bound in Theorem \ref{thm:lower-AQT}.

\vspace{2mm}

\subsubsection{Drift-plus-Penalty Algorithm}
In each slot $t$, the Drift-plus-Penalty algorithm observes the current network event $\omega_t$ and the queue length vector $\mathbf{Q}(t)$, and choose the following control action $\alpha_t^{DP}$:
\begin{equation}\label{eq:dpp}
\small
\begin{split}
\alpha_t^{DP} = \arg\max_{\alpha_t\in\mathcal{D}_{\omega_t}} &\sum_{i} Q_i(t)\Big(b_i(\alpha_t, \omega_t)-a_i(\alpha_t,\omega_t)\Big)\\
&+VU(\alpha_t,\omega_t),
\end{split}
\end{equation}
where $V>0$ is a parameter controlling the tradeoffs between utility regret and queue length. Note that $\sum_{i} Q_i(t)\Big(b_i(\alpha_t, \omega_t)-a_i(\alpha_t,\omega_t)\Big)$ corresponds to the drift part while $VU(\alpha_t,\omega_t)$ is the penalty part.

The control actions in the Drift-plus-Penalty algorithm can be usually decomposed into several actions. For example, in one-hop networks without routing, $\mathbf{a}(t)$ corresponds to the admitted exogenous arrival vector in slot $t$. Suppose that the utility function is in the form $U(\alpha_t,\omega_t)=\sum_i U_i(a_i(t))$. Then the Drift-plus-Penalty algorithm can be decomposed into the solutions of two sub-problems.
\begin{itemize}
\item{(\emph{Admission Control}) Choose 
\[
\mathbf{a}(t)=\arg\max_{\mathbf{a}}\sum_i \Big(VU_i(a_i)-Q_i(t)a_i\Big).
\]}
\item{(\emph{Resource Allocation and Scheduling}) Choose 
\[
\mathbf{b}(t)=\arg\max_{\mathbf{b}}\sum_i Q_i(t)b_i.
\]}
\end{itemize}
The first part is usually a convex optimization problem while the second part corresponds to the MaxWeight policy \cite{tassiulas}.

The following theorem gives the performance of the Drift-plus-Penalty algorithm in $W$-constrained networks.
\begin{theorem}\label{thm:max-AQT}
In any $W$-constrained network, the Drift-plus-Penalty algorithm with parameter $V$ achieves $O\Big(\frac{TW}{V}\Big)$ utility regret and the total queue length is $O\Big(\sqrt{T(W+V)}\Big)$.
\end{theorem}
\begin{proof}
The proof is based on the Lyapunov drift analysis. However, instead of considering the one-slot drift as in the traditional stochastic analysis, we find upper bounds on the $W$-slot drift-plus-penalty term and make sample-path arguments. See Appendix \ref{ap:max-AQT} for details.
\end{proof}

\vspace{1mm}

There are several important observations about Theorem \ref{thm:max-AQT}. First, if parameter $V$ is set appropriately, then sublinear utility regret and sublinear queue length can be simultaneously achieved by the Drift-plus-Penalty algorithm in $W$-constrained networks as long as $W=o(T)$.  For example, if $W=\Theta(T^{1\slash 2})$, then setting $V=\Theta(T^{3\slash 4})$ yields the utility regret of $O(T^{3\slash 4})$ and the total queue length of $O(T^{7\slash 8})$.

Noticing that sublinear utility regret and sublinear queue length cannot be achieved simultaneously by any causal policy if $W=\Omega(T)$ (Theorem \ref{thm:lower-AQT}), we have the following corollary.
\begin{corollary}
Under the $W$-constrained adversary model, sublinear utility regret and sublinear queue length are simultaneously achievable if and only if $W=o(T)$.
\end{corollary}

Second, the performance of the Drift-plus-Penalty algorithm could be much worst than the lower bound in Theorem \ref{thm:lower-AQT}. For example, if $W=\Theta(T^{1\slash 2})$, then one of the tradeoffs implied by the lower bound is that the utility regret is $\Theta(T^{1\slash 2})$ and the total queue length is also $\Theta(T^{1\slash 2})$, which is not achievable by the Drift-plus-Penalty algorithm. In the next section, we develop an algorithm that has a better performance and attains the lower bound.

\vspace{2mm}

\subsubsection{Tracking Algorithm}\label{sec:alg-tracking-W}
The tradeoff bounds achieved by the Drift-plus-Penalty algorithm is relatively loose as compared to the lower bound in Theorem \ref{thm:lower-AQT}. In this section, we develop the Tracking Algorithm that has a better performance and attains the lower bound in Theorem \ref{thm:lower-AQT}.

The original Tracking Algorithm was proposed in \cite{andrews-zhang-1, andrews-zhang-2} to solve a scheduling problem under the Adversarial Queueing Theory model. However, it only works for a very specific network model: (i) the network has to be single-hop where the arrival vector  is independent of the control action, and (ii) the control action has to satisfy the primary interference constraints, i.e., only one link incident on the same node can be activated in each slot. Next, we extend the original Tracking Algorithm to accommodate the general network model considered in this paper.

Let $\Omega$ be the set of all possible network events that could happen in each slot. In order for the Tracking Algorithm to work, the cardinality of $\Omega$ has to be finite (otherwise it could be discretized into a finite set as in \cite{andrews-zhang-1}). For example, in a single-hop network, suppose each network event $\omega_t$ corresponds to a couple $(\mathbf{A}(t),\mathbf{S}(t))$ where $\mathbf{A}(t)$ is a vector of exogenous packet arrivals in slot $t$ and $\mathbf{S}(t)$ a vector of link states in slot $t$. For any link $i$ and time $t$, assume that $0\le A_i(t)\le B$ and $A_i(t)$ is an integer, and  each link only has a finite number of $S$ states. Then $|\Omega|=(S B)^N$. 

The Tracking Algorithm is given in Algorithm \ref{alg:tracking}. It maintains an action queue $\mathcal{Q}_{\omega}$ for each type of network events $\omega\in\Omega$. The action queue $\mathcal{Q}_{\omega}$ stores the optimal actions that the Tracking Algorithm should have taken when network event $\omega$ occurred. Note that the sequence of optimal control actions cannot be calculated
online but can be calculated  every $W$ slots due to the window structure \eqref{eq:AQT}. In the Tracking Algorihtm, the sequence of optimal actions during each window are added to the action queues in batch at the end of this window (steps \ref{step:update}-\ref{step:tracking1}). Here, the  optimal actions during a window $[t-W+1, t]$ corresponds to any optimal solution to \eqref{eq:w-opt} (which is also a part of the optimal solution to \textbf{NUM}). In each slot $t$, the Tracking Algorithm first observes the current network event $\omega_t=\omega$. If the corresponding action queue $\mathcal{Q}_{\omega}$ is not empty (i.e., there are some actions we should have taken but have not taken yet), the algorithm just sets the control action as the first action in the action queue $\mathcal{Q}_{\omega}$, and the action is removed from the action queue $\mathcal{Q}_{\omega}$ (steps \ref{step:observe}-\ref{step:greedy}). If the action queue is empty, the algorithm may take any feasible action. In our analysis, we assume that no action is taken when the action queue is empty.

\vspace{-5mm}

\begin{equation}\label{eq:w-opt}
\small
\begin{split}
\max \quad &\sum_{\tau=t-W+1}^{t} U(\alpha_\tau, \omega_\tau)\\
\text{s.t.}\quad & \sum_{\tau=t-W+1}^{t}a_i(\alpha_\tau, \omega_\tau)\le \sum_{\tau=t-W+1}^{t}b_i(\alpha_\tau, \omega_\tau),~\forall i\\
& \alpha_\tau\in\mathcal{D}_{\omega_\tau},~\forall \tau.
\end{split}
\end{equation}

\begin{algorithm}
\caption{Tracking Algorithm (TA)}\label{alg:tracking}
\begin{algorithmic}[1]
\STATE Initialize $\mathcal{Q}_{\omega} =\emptyset$ for each $\omega\in \Omega$.
\FOR{$t=0,\cdots,T-1$}
\STATE Observe the current network event $\omega_t=\omega$.\label{step:observe}
\IF {action queue $\mathcal{Q}_{\omega}$ is not empty}
\STATE Choose the control action $\alpha^{TA}_t$ as the first action in $\mathcal{Q}_{\omega}$ and remove this action from $\mathcal{Q}_{\omega}$. \label{step:greedy}
\ENDIF
\IF{$\mod(t,W)=W-1$}
\STATE Compute the sequence of optimal control actions $\{\alpha^*_{\tau}\}_{\tau=t-W+1}^{t}$ in the past window $[t-W+1,~t]$, which is any optimal solution to \eqref{eq:w-opt}.
\label{step:update}
\STATE For each slot $\tau$ in the past window $[t-W+1,~t]$, enqueue the computed optimal action $\alpha^*_{\tau}$ into the action queue $\mathcal{Q}_{\omega_{\tau}}$, where $\omega_{\tau}$ is the network event occurring in slot $\tau$.\label{step:tracking1}
\ENDIF
\ENDFOR
\end{algorithmic}
\end{algorithm}

The following theorem gives the tradeoff between utility regret and queue length achieved by the Tracking Algorithm under the $W$-constrained adversary model.
\begin{theorem}\label{thm:tracking-AQT}
In any $W$-constrained network, the Tracking Algorithm achieves $O(W)$ utility regret and the total queue length is $O(W)$.
\end{theorem}
\begin{proof}
Since the Tracking Algorithm updates the optimal actions every $W$ slots and replays these actions whenever possible, the number of unfulfilled actions in any action queue is at most $W$. Thus, the performance gap between the Tracking Algorithm and optimal policy is also $O(W)$. See Appendix \ref{ap:tracking-AQT} for details.
\end{proof}

\vspace{1mm}

There are several important observations about Theorem \ref{thm:tracking-AQT}.  First, under the $W$-constrained adversary model, sublinear utility regret and sublinear queue length can be simultaneously achieved by the Tracking Algorithm as long as $W=o(T)$. Moreover, the tradeoff achieved by the Tracking Algorithm is better than that of the Drift-plus-Penlaty algorithm, in terms of their dependence on $W$ and $T$. For example, if $W=\Theta(T^{1\slash 2})$, the Tracking Algorithm can achieve $O(T^{1\slash 2})$ utility regret and $O(T^{1\slash 2})$ total queue length, while such a tradeoff is not attainable by the Drift-plus-Penalty algorithm.

Second, the Tracking Algorithm asymptotically achieves the lower bound in Theorem \ref{thm:lower-AQT} in the sense that it ensures that $\mathcal{R}_T\Big(\{\omega_0,\cdots,\omega_{T-1}\}\Big)+\sum_i Q_i(T)=O(W)$ for any $\{\omega_0,\cdots,\omega_{T-1}\}\in\mathcal{W}_T$. As a result, the Tracking Algorithm asymptotically achieves the optimal tradeoff between utility regret and queue length w.r.t. $W$ and $T$.

Third, the Tracking Algorithm needs to maintain a virtual queue for each type of network events while the size of the network event space $\Omega$ may be exponential in the number of users $N$. As a result, the Tracking Algorithm may not be a practical algorithm. The purpose of presenting the Tracking Algorithm is to demonstrate that the lower bound in Theorem \ref{thm:lower-AQT} could be asymptotically achieved by a causal policy. Note that Andrews and Zhang \cite{andrews-zhang-1} proposed a method to get rid of the exponential dependence on $N$, at the expense of much more involved algorithm.

Finally, the Tracking Algorithm described in Algorithm \ref{alg:tracking} only achieves one point in the tradeoff curve since it only tracks the optimal solution to \textbf{NUM}. One approach to enable tunable tradeoffs is to relax the optimization problem \eqref{eq:w-opt}. For example, the first constraint in \eqref{eq:w-opt} can be modified to
\[
\sum_{\tau=t-W+1}^{t}a_i(\alpha_\tau, \omega_\tau)\le \sum_{\tau=t-W+1}^{t}b_i(\alpha_\tau, \omega_\tau) + V,
\]
for some parameter $V$. Clearly, by tuning the value of $V$, the optimal solution  to \eqref{eq:w-opt} (denoted by $\{\alpha^V_t\}_{t=0}^{T-1}$) can achieve different tradeoffs. By tracking the solution $\{\alpha^V_t\}_{t=0}^{T-1}$, the Tracking Algorithm can achieve tunable tradeoffs. The analysis of the tunable Tracking Algorithm is similar to the proof of Theorem \ref{thm:tracking-AQT} but requires more specific assumptions on the utility function, and is omitted due to space constraints.

Note that the above Tracking Algorithm requires $W$ as a parameter. We discuss how to properly select the value of $W$ in Section \ref{sec:sensitivity}.

\section{$V_T$-Constrained Adversary Model}\label{sec:VT}
The aforementioned $W$-constrained model is  relatively restrictive, where the stringent constraints \eqref{eq:AQT} have to be satisfied for every window of $W$ slots. In this section, we consider a general adversary model where the window constraints \eqref{eq:AQT} are relaxed.

The new adversary model is parameterized by the inherent variation in the sequence of network events, which is measured as follows. Given a sequence of network events $\omega_0,\cdots,\omega_{T-1}$ and a (possibly non-causal) policy, we define
\[
V^\pi\Big(\{\omega_0,\cdots,\omega_{T-1}\}\Big)=\max_{t\le T}\sum_i Q_i^\pi(t).
\]
The above function measures the peak queue length achieved by policy $\pi$ during its sample path. We further define
$
V^*\Big(\{\omega_0,\cdots,\omega_{T-1}\}\Big)
$
to be the peak queue length during the sample path of the optimal solution to \textbf{NUM} under the sequence of network events $\omega_0,\cdots,\omega_{T-1}$. If there are multiple optimal solutions to \textbf{NUM}, then the one with the smallest value of $V^*(\cdot)$ is considered. Note that $V^*(\cdot)$ only depends on $\{\omega_0,\cdots,\omega_{T-1}\}$ and measures the inherent variations in the sequence of network events.

Now we define the notion of $V_T$-constrained network dynamics where the value of $V^*(\cdot)$ is constrained by some budget $V_T$.
\begin{definition}[$V_T$-Constrained Dynamics]
Given some $V_T\in [0,NTB]$, a sequence of network events $\omega_0,\cdots,\omega_{T-1}$ is $V_T$-constrained if
\[
V^*\Big(\{\omega_0,\cdots,\omega_{T-1}\}\Big)\le V_T.
\]
\end{definition}
\noindent Any network satisfying the above is called a  \textbf{$V_T$-constrained network}. Denote by $\mathcal{V}_T$ the set of all possible  sequences of network events that are  $V_T$-constrained. A \textbf{$V_T$-constrained adversary} can only select the sequence of network events from the set $\mathcal{V}_T$. 

Note that we restrict the range of $V_T$ to $[0,NTB]$ since the peak queue length within $T$ slots is at most $NTB$. Any larger value of $V_T$  has the same effect as $V_T=NTB$. Note also that the larger $V_T$ is, the more variations the network could have. By varying the value of $V_T$ from 0 to  $NTB$, the above $V_T$-constrained adversary model covers a wide range of adversarial settings: from a strictly constrained adversary ($V_T=0$, i.e., the arrivals should not exceed the services for each queue in every slot) to a completely unconstrained adversary ($V_T=NTB$).

In the following, we first provide a lower bound on the tradeoffs between utility regret and queue length under the $V_T$-constrained adversary model in Section \ref{sec:lower-vt} and then analyze the performance of the Drift-plus-Penalty policy and the Tracking Algorithm in Section \ref{sec:alg-vt}.
\subsection{Lower Bound on the Tradeoffs}\label{sec:lower-vt}
The following theorem provides a lower bound on the tradeoffs between utility regret and queue length under the $V_T$-constrained adversary model.
\begin{theorem}\label{thm:lower-general}
For any causal policy $\pi$, there exists a sequence of network events $\{\omega_0,\cdots,\omega_{T-1}\}\in \mathcal{V}_T$ such that 
\[
\mathcal{R}^\pi_T\Big(\{\omega_0,\cdots,\omega_{T-1}\}\Big)+c\sum_i Q_i^\pi(T)\ge c'V_T,
\]
where $c'>0$ is some constant.
\end{theorem}
\begin{proof}
The proof is the same as that for Theorem \ref{thm:lower-AQT} except that we replace $W$ with $V_T$, thus omitted for brevity.
\end{proof}

\vspace{1mm}

Theorem \ref{thm:lower-general} shows that if $V_T=\Omega(T)$, then no causal policy can simultaneously achieve sublinear utility regret and sublinear queue length under the $V_T$-constrained adversary model. On the other hand, if $V_T=o(T)$, there \emph{might} exist some causal policy that attains sublinear utility regret and sublinear queue length simultaneously, which we investigate in Section \ref{sec:alg-vt}.
\subsection{Algorithm Performance in $V_T$-Constrained Networks}\label{sec:alg-vt}
In this section, we analyze the tradeoffs between utility regret and queue length achieved by two algorithms in $V_T$-constrained networks: the Drift-plus-Penalty algorithm and the Tracking Algorithm. In particular, we show that both algorithms simultaneously achieve sublinear utility regret and sublinear queue length if $V_T=o(T)$.

\vspace{2mm}
\subsubsection{Drift-plus-Penalty Algorithm}
The Drift-plus-Penalty algorithm discussed in Section \ref{sec:alg-W} can be directly applied to the $V_T$-constrained setting. The following theorem gives the tradeoffs between utility regret and  queue length achieved by the Drift-plus-Penalty algorithm under the $V_T$-constrained adversary model.

\begin{theorem}\label{thm:max-general}
In any $V_T$-constrained network, the Drift-plus-Penalty algorithm with parameter $V>0$ achieves $O\Big(\frac{V_T^{2\slash 3}T^{4\slash 3}}{V}+\frac{V_T^{1\slash 3}T^{7\slash 6}}{V^{1\slash 2}}\Big)$ utility regret and the total queue length is $O\Big(V_T^{1\slash 3} T^{2\slash 3}+T^{1\slash 2}V^{1\slash 2}\Big)$.
\end{theorem}
\begin{proof}
We first divide the time horizon into frames of $W$ slots. Then we apply the analysis used in the $W$-constrained adversary model and derive bounds on the $W$-slot drift-plus-penalty term, which further leads to upper bounds on utility regret and queue length. The value of $W$ is carefully chosen to optimize these bounds. See Appendix \ref{ap:max-general} for details. 
\end{proof}

\vspace{1mm}

There are several observations about Theorem \ref{thm:max-general}. First, the Drift-plus-Penalty algorithm achieves sublinear utility regret and sublinear queue length under the $V_T$-constrained adversary model whenever $V_T=o(T)$. For example,  if $V_T=\Theta(T^{1\slash 2})$ and we set $V=\Theta(T^{4\slash 5})$, then the utility regret and the total queue length are both $O(T^{11\slash 12})$. Notice that sublinear utility regret and sublinear queue length cannot be simultaneously  achieved by any causal policy if $V_T=\Omega(T)$ (Theorem \ref{thm:lower-general}). We have the following corollary.
\begin{corollary}
Under the $V_T$-constrained adversary model, sublinear utility regret and sublinear queue length are simultaneously achievable  if and only if $V_T=o(T)$.
\end{corollary}

\noindent Second, the Drift-plus-Penalty algorithm does not attain the lower bound in Theorem \ref{thm:lower-general}. For example, if $V_T=\Theta(T^{1\slash 2})$, one of the tradeoffs implied by the lower bound is that the utility regret is $\Theta(T^{1\slash 2})$ and the total queue length is also $\Theta(T^{1\slash 2})$, which is not achievable by the Drift-plus-Penalty algorithm. In fact, although  the Drift-plus-Penalty algorithm can achieve sublinear utility regret and sublinear queue length, the tradeoff bound in Theorem \ref{thm:max-general} is relatively loose. In the next section, we show that the Tracking Algorithm can achieve a better tradeoff bound than the Drift-plus-Penalty algorithm.

\vspace{2mm}

\subsubsection{Tracking Algorithm}
The Tracking Algorithm introduced under the $W$-constrained adversary model requires that the window constraints \eqref{eq:AQT} be satisfied for some window size $W$. However, there might be no window structure under the $V_T$-constrained adversary model and thus the Tracking Algorithm cannot be directly applied in $V_T$-constrained networks. We slightly modify the Tracking Algorithm in two aspects. First, the window size $W$ is set to be $W=\sqrt{TV_T}$ under the $V_T$-constrained adversary model. Second, in step \ref{step:update} of the original Tracking Algorithm, the optimization problem \eqref{eq:w-opt} is modified to be
\begin{equation}\label{eq:track-shed}
\small
\begin{split}
\max & \sum_{\tau=t-W+1}^{t} U(\alpha_\tau, \omega_\tau)\\
\text{s.t.} & \sum_{\tau=t-W+1}^{t}a_i(\alpha_\tau, \omega_\tau)\le \sum_{\tau=t-W+1}^{t}b_i(\alpha_\tau, \omega_\tau)+V_T,~\forall i\\
& \alpha_\tau\in\mathcal{D}_{\omega_\tau},~\forall \tau.
\end{split}
\end{equation}
In particular, the first constraint in \eqref{eq:w-opt} is relaxed by allowing some bursts up to $V_T$.
Note that by the definition of $V_T$-constrained networks, the optimal solution to $\textbf{NUM}$ is also a feasible solution to \eqref{eq:track-shed}.
Under the above setting, the utility regret and the total queue length achieved by the Tracking Algorithm in $V_T$-constrained networks is given by the following theorem\footnote{As is discussed in Section \ref{sec:alg-tracking-W}, the set of possible network events should be finite in order for the Tracking Algorithm to work.}.
\begin{theorem}\label{thm:tracking-general}
Under the $V_T$-constrained adversary model, the Tracking Algorithm achieves $O(\sqrt{TV_T})$ utility regret and the total queue length is $O(\sqrt{TV_T})$.
\end{theorem}
\begin{proof}
The proof is similar to the analysis under the $W$-constrained adversary model, except that an additional $V_T$ terms is added in the first constraint of \eqref{eq:track-shed}. See Appendix \ref{ap:tracking-general} for details.
\end{proof}

\vspace{1mm}

There are several important observations about Theorem \ref{thm:tracking-general}. First, the Tracking Algorithm can simultaneously achieve sublinear utility regret and sublinear queue length whenever $V_T=o(T)$. 
Second, the performance of the Tracking Algorithm is better than that under the Drift-plus-Penalty algorithm in $V_T$-constrained networks. For example, if we set $W=\Theta(\sqrt{V_T T})$ and $V_T=\Theta(\sqrt{T})$, then the Tracking Algorithm achieves $O(T^{3\slash 4})$ utility regret and $O(T^{3\slash 4})$ queue length, which is not achievable by the Drift-plus-Penalty algorithm.
Finally, the Tracking Algorithm does not attain the tradeoff lower bound in Theorem \ref{thm:lower-general}.  Thus, finding a causal policy that can close the gap remains an open problem.

Note that the above Tracking Algorithm requires $V_T$ as a parameter. We discuss how to properly select the value of $V_T$ in Section \ref{sec:sensitivity}.


\subsection{Discussions}\label{discussions}
\subsubsection{Relationship between Adversary Models}\label{sec:comp}
The $V_T$-constrained adversary model generalizes the $W$-constrained adversary model: any sequence of network events that are $W$-constrained must also be $V_T$-constrained with $V_T=O(W)$ due to the window structure (note that the peak queue length under the optimal policy is at most $NWB$). The analysis in the $V_T$-constrained adversary model also gives a  more general condition for sublinear utility regret and sublinear queue length.
\subsubsection{Choosing Parameters for Tracking Algorithm}\label{sec:sensitivity}
Note that the Tracking Algorithm requires $V_T$ as a parameter. Unfortunately, in practice, it is impossible to know the precise value of $V_T$ for a given network in advance.
To alleviate this issue, we can search for the correct value of $V_T$. Note that the range for $V_T$ is $[0,NBT]$. Then one may perform binary search to find the correct value of $V_T$ by running the Tracking algorithm with different values of $V_T$ over multiple episodes within the time horizon (e.g., if the time horizon is $T=10^5$ slots, then one episode could be $10^3$ slots). Similar techniques can be applied if the Tracking Algorithm is used in $W$-constrained networks where the value of $W$ is required as input parameters.

\section{Simulations}\label{sec:sim}
In this section, we empirically validate the theoretical bounds derived in this paper and compare the performance of the Drift-plus-Penalty  and the Tracking Algorithm.

\begin{figure*}[ht!]
\subfigure[Queue Length ($W=\Theta(\sqrt{T})$) ]{\label{fig:adaptive-scaling-1}\includegraphics[width=43mm,height=33mm]{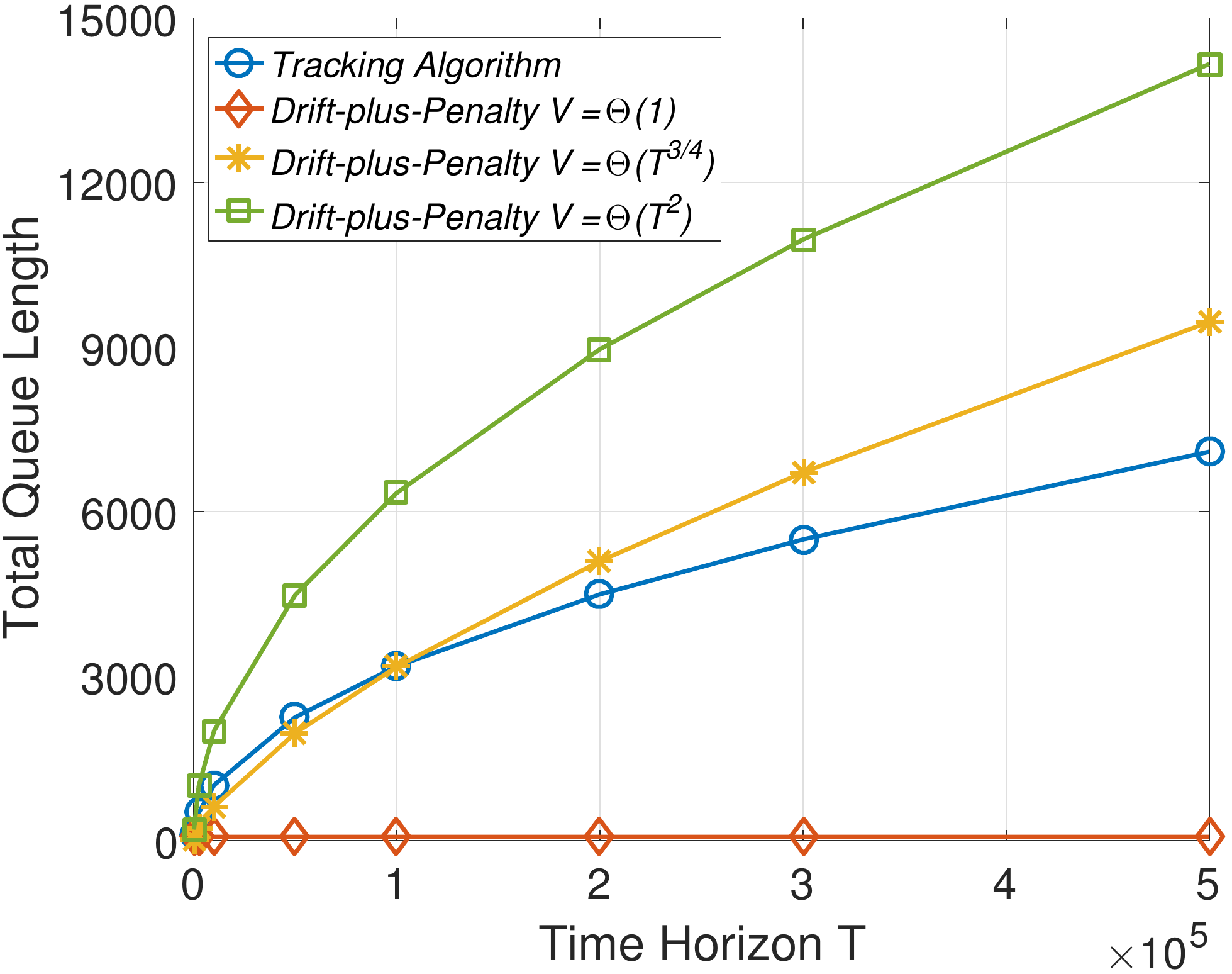}}
\subfigure[Utility Regret ($W=\Theta(\sqrt{T})$)]{\label{fig:adaptive-scaling-2}\includegraphics[width=43mm,height=35mm]{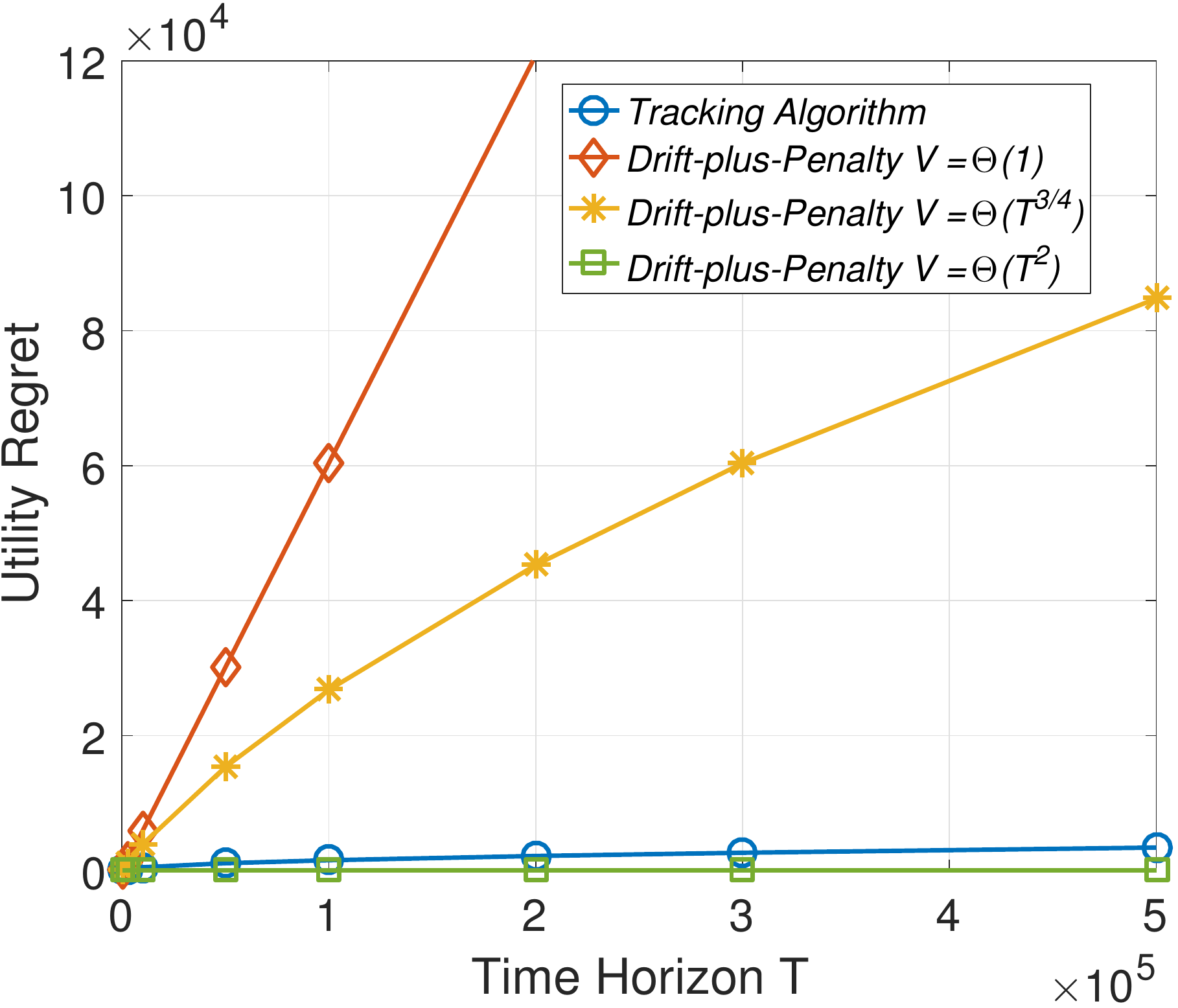}}
\subfigure[Queue Length ($W=\Theta(T)$) ]{\label{fig:adaptive-scaling-3}\includegraphics[width=43mm,height=35mm]{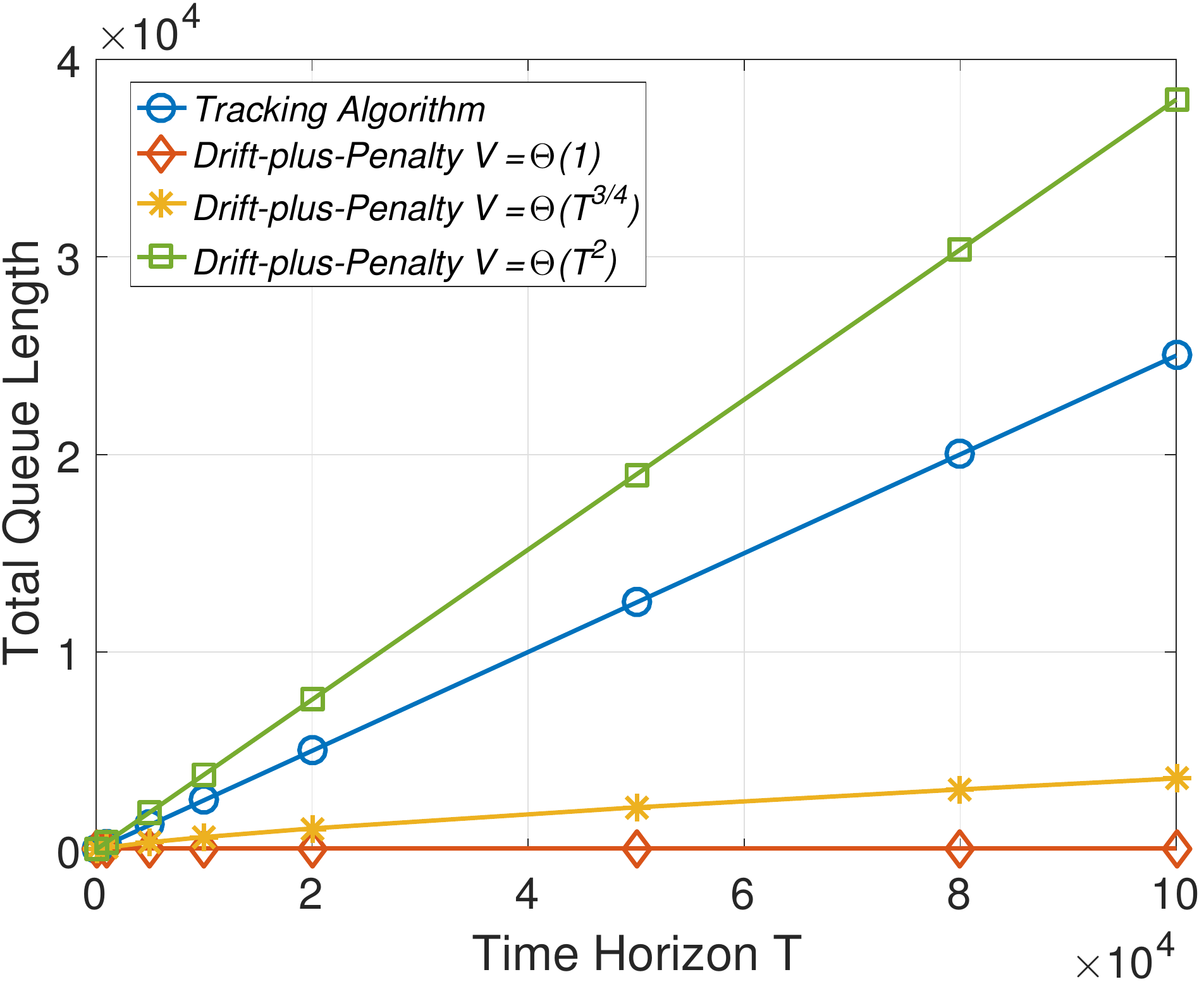}}
\subfigure[Utility Regret ($W=\Theta(T)$)]{\label{fig:adaptive-scaling-4}\includegraphics[width=43mm,height=35mm]{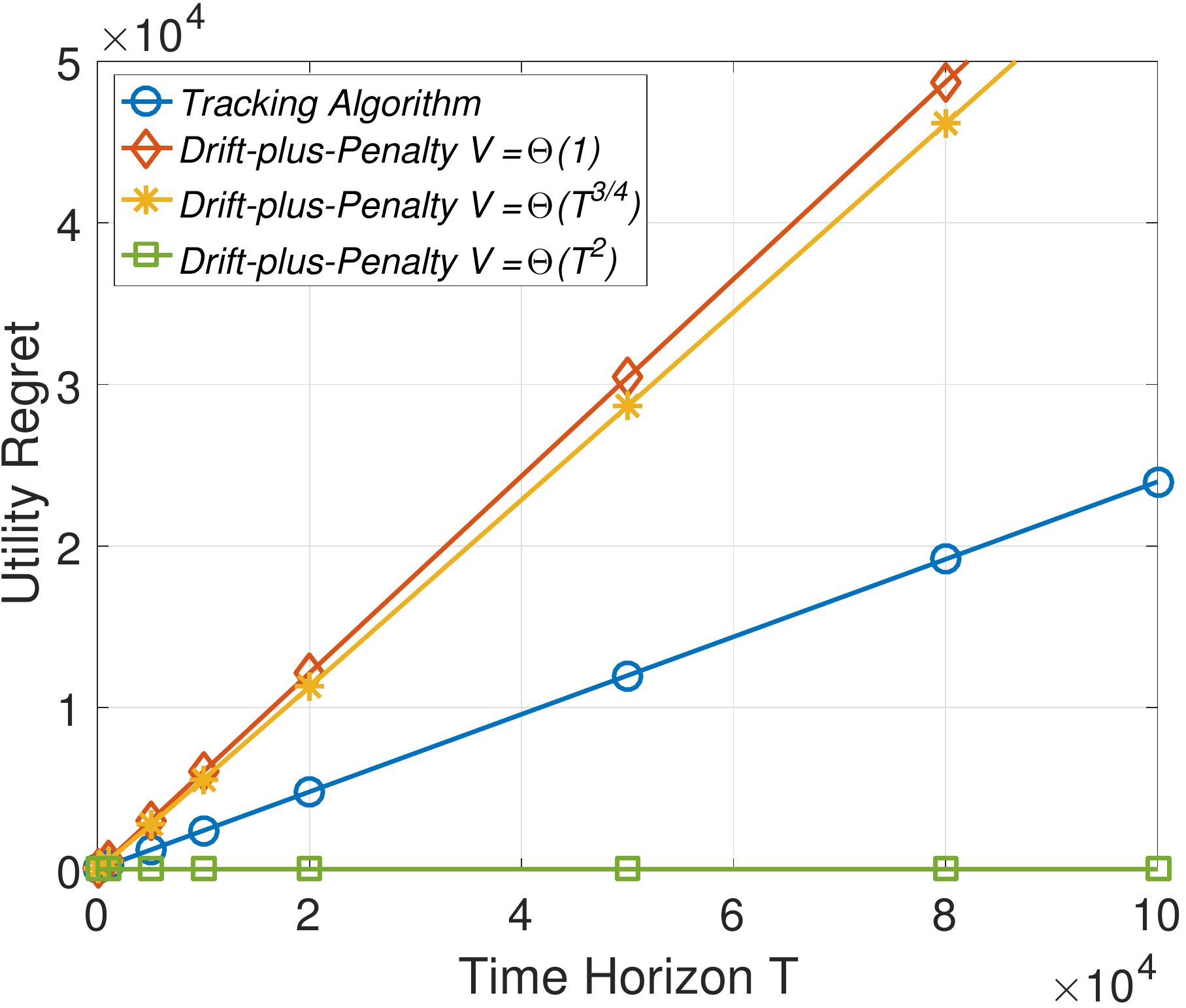}}
\vspace{-2mm}
\caption{Growth of total queue length and utility regret with the time horizon $T$ under an adaptive $W$-adversary.}
\label{fig:scaling}\vspace{-2mm}
\end{figure*}

In our simulations, we consider a single-hop network with $N=2$ users. In each slot $t$, the central controller observes the current network event $\omega_t=\big(\mathbf{A}(t),\mathbf{S}(t)\big)$, where $\mathbf{A}(t)$ is the exogenous arrival vector and $\mathbf{S}(t)$ is the channel rate vector for each link in slot $t$. Then the controller makes an admission control and a scheduling decision. The constraint on the admission control action is $0\le a_i(t)\le A_i(t)$ for each link $i$, and the constraint on the scheduling decision is that at most one of the links can be served in each slot. The network utility is $U(\alpha_t,\omega_t)=\sum_i \log\big(1+a_i(t)\big)$ (proportional fairness). We consider a scenario where the channel rate vector in each slot is controlled by an adaptive adversary. Time is divided into frames of $W$ slots. In the first $\lceil W\slash 2\rceil$ slots of each frame, the exogenous arrivals to each user are 10 packets/slot and the channel rate for each user is also 10 packets/slot. In the remaining slots of each frame, there are no exogenous arrivals to both users while the channel rate is zero for the user with a longer queue and 10 packets/slot for the other user. If the two users have the same queue length, ties are broken randomly. Such a scenario is similar to the one that we use to prove the tradeoff lower bound under the $W$-constrained adversary model (see the proof of Theorems \ref{thm:lower-AQT}), and it has been shown that this is a $W$-constrained adversary (and also a $V_T$-constrained adversary with $V_T=5W$).

Figure \ref{fig:scaling} illustrates the growth of the total queue length and the utility regret with the time horizon $T$ under the Drift-plus-Penalty algorithm (with different values of $V$) and the Tracking algorithms. First, when $W=\Theta(\sqrt{T})$, the Drift-plus-Penalty can simultaneously achieve sublinear utility regret and sublinear queue length, if the parameter $V$ is set appropriately (for example, $V=\Theta(T^{3\slash 4})$). Note that setting $V$ to some very large value (e.g., $V=\Theta(T^2)$) still achieves sublinear utility regret and sublinear queue length, though the theoretical bound on queue length (see Theorem \ref{thm:max-AQT}) is at least linear in $T$ when $V=\Omega(T)$, which shows that the performance upper bound is not tight in this scenario. The Tracking Algorithm also simultaneously achieves sublinear utility regret and sublinear queue length when $W=\Theta(\sqrt{T})$. However, when $W=\Theta(T)$, both algorithms fail to achieve desirable performance: either the utility regret or the queue length grows linearly with $T$. In fact, the lower bound in Theorem \ref{thm:lower-AQT} shows that no causal policy can achieve both sublinear utility regret and sublinear queue length if $W=\Theta(T)$.

Figure \ref{fig:tradeoff} shows the tradeoffs between utility regret and queue length under the Drift-plus-Penalty algorithm and the Tracking Algorithm, where we fix the time horizon to be $T=10^4$ slots and the window size $W=\Theta(\sqrt{T})$. Note that for the Drift-plus-Penalty algorithm, we plot a tradeoff curve (since it achieves different tradeoffs by tuning the parameter $V$), while only a single tradeoff point is plotted for the Tracking Algorithm.
It is observed that the Tracking Algorithm achieves a better tradeoff point that is not achievable by the Drift-plus-Penalty algorithm. In addition, the theoretical lower bound for any causal policy (Theorem \ref{thm:lower-AQT}) and the theoretical performance upper bounds for both algorithms (Theorems \ref{thm:max-AQT} and \ref{thm:tracking-AQT}) are also validated in the figure. 
\begin{figure}[ht!]
\begin{center}
\includegraphics[width=2.8in]{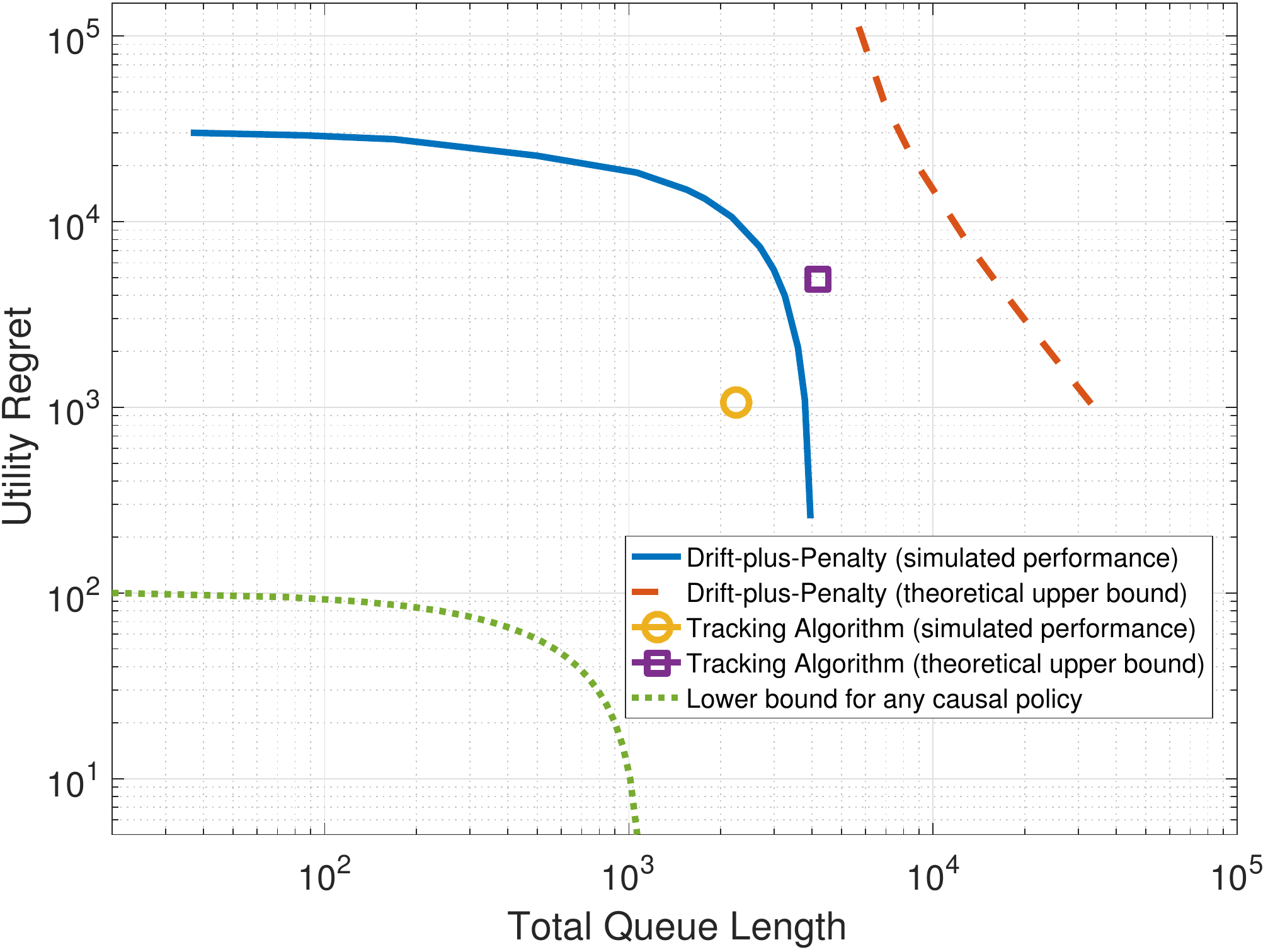}
\caption{Tradeoffs between utility and total queue length (double log scale). The time horizon is fixed to be $T=10^4$ slots and $W=\Theta(\sqrt{T})$.}
\label{fig:tradeoff}
\end{center}\vspace{-4mm}
\end{figure}

\section{Conclusions}\label{sec:conclusion}
In this paper, we focus on optimizing network utility within a finite time horizon under adversarial network models. We show that no causal policy can simultaneously achieve both sublinear utility regret and sublinear queue length if the network dynamics are unconstrained, and investigate two constrained adversary models. We first consider the restrictive $W$-constrained adversary model and then propose a more relaxed $V_T$-constrained adversary model. Lower bounds on the tradeoffs between utility regret and queue length are derived under the two adversary models, and  the  performance of two control policies is analyzed, i.e., the Drift-plus-Penalty algorithm and the Tracking Algorithm. It is shown that the Tracking Algorithm asymptotically attains the optimal tradeoffs under the $W$-constrained adversary model and that the Tracking Algorithm has a better tradeoff bound than that of the Drift-plus-Penalty


\newpage

\appendices
\section{Proof to Theorem \ref{thm:impossible}}\label{ap:impossible}
We prove this theorem by constructing a sequence of network events $\omega_0,\cdots,\omega_{T-1}$ such that either utility regret or total queue length grows at least linearly with the time horizon $T$. Consider a single-hop network with 2 links. In each slot $t$, the central controller observes the current network event $\omega_t=\big(\mathbf{A}(t),\mathbf{S}(t)\big)$, where $\mathbf{A}(t)$ is the exogenous arrival vector and $\mathbf{S}(t)$ is the channel rate vector for each link in slot $t$. Then the controller makes an admission control and a scheduling decision. The constraint on the admission control action is $0\le a_i(t)\le A_i(t)$ for each link $i$, and the constraint on the scheduling decision is that at most one of the links can be served in each slot. The network utility is a function of the admitted traffic vector $\mathbf{a}(t)$, i.e., $U(\alpha_t,\omega_t)=U(\mathbf{a}(t))=\sum_i U_i(a_i(t))$, where $U_i(x)$ is convex and strictly increasing in $x$. In particular, any subderivative of $U_i(x)$ over the range $x\in[0,B]$ is lower bounded by some constant $c>0$. Typical examples of such utility functions are $U(\mathbf{a}(t))=\sum_i a_i(t)$ (total throughput) and $U(\mathbf{a}(t))=\sum_i \log\big(a_i(t)\big)$ (proportional fairness).

Without loss generality, assume that the time horizon $T$ is an even number. The exogenous arrivals and channel rates in the first $T\slash 2$ slots are
\[
A_1(t)=A_2(t)=2,~~S_1(t)=S_2(t)=2,~\forall t=0,\cdots,\frac{T}{2}-1.
\]
For any causal policy $\pi$, let $B^\pi_1$ and $B^\pi_2$ be the number of packets cleared over link 1 and 2 during the first $T\slash 2$ slots, respectively. Also let $A^\pi_1$ and $A^\pi_2$ be the number of admitted packets to link 1 and link 2 during the first $T\slash 2$ slots, respectively. Then the queue length vector after the first $T\slash 2$ slots is
\[
Q^\pi_i(T\slash 2)=A^\pi_i-B^\pi_i,~i=1,2.
\]

Under the scheduling constraint, the total number of packets that can be cleared in the first $T\slash 2$ slots is at most $T$. Then we have $B^\pi_1+B^\pi_2\le T$, which implies that $\min\{B^\pi_1,B^\pi_2\}\le T\slash 2$. Define $i^*\triangleq \arg\min_{i} B^\pi_i$. In the remaining $T\slash 2$ slots, the adversary can set
\[
A_{i^*}(t)=0,~S_{i^*}(t)=0,~t=T\slash 2,\cdots,T-1.
\]
For the other link (its index is denoted by $i'$), the adversary can set
\[
A_{i'}(t)=0,~S_{i'}(t)=2,~t=T\slash 2,\cdots,T-1.
\]
Since there is no capacity to clear any packet over link $i^*$ in the remaining $T \slash 2$ slots, we have
\begin{equation}\label{eq:q-T-1}
Q^{\pi}_{i^*}(T)=Q^{\pi}_{i^*}(T\slash 2)=A^\pi_{i^*}-B^\pi_{i^*}.
\end{equation}
Note that the optimal non-causal policy can admit all the exogenous traffic while keeping the total queue length $\sum_i Q^*_i(T)=0$ by serving link $i^*$ in the first $T\slash 2$ slots and serving link $i'$ in the remaining $T\slash 2$ slots. As a result, the utility regret is
\begin{equation}\label{eq:b-regret-T}
\begin{split}
\mathcal{R}^\pi_T\Big(\{\omega_0,\cdots,\omega_{T-1}\}\Big)&=\sum_{t=0}^{T-1}\Big[U(\mathbf{a}^*(t))-U(\mathbf{a}^{\pi}(t))\Big]\\
&=\sum_{t=0}^{T-1}\sum_i \Big[U_i(a_i^*(t))-U_i(a_i^{\pi}(t))\Big]\\
&\ge c \sum_{t=0}^{T-1}\sum_i \Big(a_i^*(t)-a_i^{\pi}(t)\Big)\\
&= c (2T-A^\pi_1-A^\pi_2),
\end{split}
\end{equation}
where the inequality is due to the concavity of the utility function and the fact that the subderivatives of the utility function are lower-bounded by $c>0$. The last equality holds because the total admitted traffic by the optimal policy is $\sum_{t=0}^{T-1}\sum_i a_i^*(t)=2T$ while the total admitted traffic by the causal policy $\pi$ is $\sum_{t=0}^{T-1}\sum_i a_i^{\pi}(t)=A^\pi_1+A^\pi_2$.
Then it follows that
\[
\begin{split}
&\mathcal{R}^\pi_T\Big(\{\omega_0,\cdots,\omega_{T-1}\}\Big)+c \sum_i Q^{\pi}_{i}(T) \\
\ge &\mathcal{R}^\pi_T\Big(\{\omega_0,\cdots,\omega_{T-1}\}\Big)+c Q^{\pi}_{i^*}(T)\\
\ge &c (2T-A^\pi_1-A^\pi_2 + A^\pi_{i^*}-B^\pi_{i^*})\\
=& c(2T-A^\pi_{i'}-B^\pi_{i^*})\\
\ge& c(2T-T-T\slash 2)\\
= & cT\slash 2, 
\end{split}
\]
where the second inequality is due to \eqref{eq:q-T-1} and \eqref{eq:b-regret-T}, and the last inequality holds because the total admitted traffic over link $i'$ is $A^\pi_{i'}\le T$ and the amount of cleared traffic over $i^*$  is $B^\pi_{i^*}\le T\slash 2$ by the definition of $i^*$. Therefore, it is impossible for any causal policy $\pi$ to simultaneously achieve both sublinear utility regret and sublinear queue length, otherwise $\mathcal{R}^\pi_T\Big(\{\omega_0,\cdots,\omega_{T-1}\}\Big)+c \sum_i Q^{\pi}_{i}(T)=o(T)$.

\vspace{2mm}

\noindent \textbf{Remark:} Note that the above construction requires the value of $T$. We can eliminate the dependence on the time horizon $T$ by using the standard \emph{Doubling Tricks} (see Section 2.3.1 in \cite{OCO}).

\section{Proof to Theorem \ref{thm:lower-AQT}}\label{ap:lower-AQT}
We prove this theorem by constructing a sequence of network events $\{\omega_0,\cdots,\omega_{T-1}\}\in\mathcal{W}_T$ such that the lower bound is attained. Consider the same network setting as in the proof of Theorem \ref{thm:impossible}. Without loss generality, assume that the window size $W$ is an even number. The exogenous arrivals and channel rates in the first $W\slash 2$ slots are
\[
A_1(t)=A_2(t)=2,~~S_1(t)=S_2(t)=2,~\forall t=0,\cdots,\frac{W}{2}-1.
\]
For any causal policy $\pi$, let $B^\pi_1$ and $B^\pi_2$ be the number of packets cleared over link 1 and 2 during the first $W\slash 2$ slots, respectively. Also let $A^\pi_1$ and $A^\pi_2$ be the number of admitted packets to link 1 and link 2 during the first $W\slash 2$ slots, respectively. Then the queue length vector after the first $W\slash 2$ slots is
\[
Q^\pi_i(W\slash 2)=A^\pi_i-B^\pi_i,~i=1,2.
\]

Under the scheduling constraint, the total number of packets that can be cleared in the first $W\slash 2$ slots is at most $W$. Then we have $B^\pi_1+B^\pi_2\le W$, which implies that $\min\{B^\pi_1,B^\pi_2\}\le W\slash 2$. Define $i^*\triangleq \arg\min_{i} B^\pi_i$. In the remaining $T-W\slash 2$ slots, the adversary can set
\[
A_{i^*}(t)=0,~S_{i^*}(t)=0,~t=W\slash 2,\cdots,T-1.
\]
For the other link (its index is denoted by $i'$), the adversary can set
\[
A_{i'}(t)=0,~S_{i'}(t)=2,~t=W\slash 2,\cdots,T-1.
\]
Since there is no capacity to clear any packet over link $i^*$ in the remaining $T-W \slash 2$ slots, we have
\begin{equation}\label{eq:q-W-1}
Q^{\pi}_{i^*}(T)=Q^{\pi}_{i^*}(W\slash 2)=A^\pi_{i^*}-B^\pi_{i^*}.
\end{equation}
Note that the optimal non-causal policy can admit all the exogenous traffic while satisfying the window constraints \eqref{eq:AQT} by serving link $i^*$ in the first $W\slash 2$ slots and serving link $i'$ in the remaining $T-W\slash 2$ slots. As a result, the utility regret is
\begin{equation}\label{eq:b-regret-W}
\begin{split}
\mathcal{R}^\pi_T\Big(\{\omega_0,\cdots,\omega_{T-1}\}\Big)&=\sum_{t=0}^{T-1}\Big[U(\mathbf{a}^*(t))-U(\mathbf{a}^{\pi}(t))\Big]\\
&=\sum_{t=0}^{T-1}\sum_i \Big[U_i(a_i^*(t))-U_i(a_i^{\pi}(t))\Big]\\
&\ge c \sum_{t=0}^{T-1}\sum_i \Big(a_i^*(t)-a_i^{\pi}(t)\Big)\\
&= c (2W-A^\pi_1-A^\pi_2),
\end{split}
\end{equation}
where the inequality is due to the concavity of the utility function and the fact that the subderivatives of the utility function are lower-bounded by $c>0$. The last equality holds because the total admitted traffic by the optimal policy is $\sum_{t=0}^{T-1}\sum_i a_i^*(t)=2W$ while the total admitted traffic by the causal policy $\pi$ is $\sum_{t=0}^{T-1}\sum_i a_i^{\pi}(t)=A^\pi_1+A^\pi_2$.
Then it follows that
\[
\begin{split}
&\mathcal{R}^\pi_T\Big(\{\omega_0,\cdots,\omega_{T-1}\}\Big)+c \sum_i Q^{\pi}_{i}(T) \\
\ge &\mathcal{R}^\pi_T\Big(\{\omega_0,\cdots,\omega_{T-1}\}\Big)+c Q^{\pi}_{i^*}(T)\\
\ge &c (2W-A^\pi_1-A^\pi_2 + A^\pi_{i^*}-B^\pi_{i^*})\\
=& c(2W-A^\pi_{i'}-B^\pi_{i^*})\\
\ge& c(2W-W-W\slash 2)\\
= & cW\slash 2\triangleq c'W,
\end{split}
\]
where the second inequality is due to \eqref{eq:q-W-1} and \eqref{eq:b-regret-W}, and the last inequality holds because the total admitted traffic over link $i'$ is $A^\pi_{i'}\le W$ and the amount of cleared traffic over $i^*$ is $B^\pi_{i^*}\le W\slash 2$ by the definition of $i^*$. This completes the proof.
\section{Proof to Theorem \ref{thm:max-AQT}}\label{ap:max-AQT}
Let $\alpha_t$, $\mathbf{Q}(t)$, $\mathbf{a}(t)$ and $\mathbf{b}(t)$ be the control action, the queue length vector, the controlled arrival vector and the service vector in slot $t$ under the Drift-plus-Penalty algorithm, respectively. Also define the potential function
\[
\Phi(\mathbf{Q}(t)) = \frac{1}{2}\sum_{i\in\mathcal{N}} Q_i^2(t).
\]
We first provide an upper bound for the $W$-slot drift $\Phi(\mathbf{Q}(t+W))-\Phi(\mathbf{Q}(t))$.
\begin{lemma}\label{lm:drift}
The $W$-slot drift satisfies
\[
\begin{split}
&\Phi(\mathbf{Q}(t+W))-\Phi(\mathbf{Q}(t))\\
\le &\sum_{\tau=t}^{t+W-1} \sum_{i=1}^N Q_i(\tau) \Big(a_i(\tau) - b_i(\tau)\Big) + 4NB^2 W^2.
\end{split}
\]
\end{lemma}
\begin{proof}
For any $i\in \mathcal{N}$,  if $Q_i(t) \ge \sum_{\tau=t}^{t+W-1} b_i(\tau)$, then
\[
Q^2_i(t+W) = \Big[Q_i(t) + \sum_{\tau=t}^{t+W-1}a_i(\tau) - \sum_{\tau=t}^{t+W-1}b_i(\tau)\Big]^2.
\]
If $Q_i(t) < \sum_{\tau=t}^{t+W-1} b_i(\tau)$, then
\[
\begin{split}
Q^2_i(t+W) &\le \Big[Q_i(t) + \sum_{\tau=t}^{t+W-1} a_i(\tau)\Big]^2 \\
&<  \Big[\sum_{\tau=t}^{t+W-1} b_i(\tau) + \sum_{\tau=t}^{t+W-1}a_i(\tau)\Big]^2.
\end{split}
\]
Thus, in any case, we have
\[
\begin{split}
Q^2_i(t+W)&\le  \Big[Q_i(t) + \sum_{\tau=t}^{t+W-1} a_i(\tau) - \sum_{\tau=t}^{t+W-1} b_i(\tau)\Big]^2\\
&+\Big[\sum_{\tau=t}^{t+W-1} b_i(\tau) + \sum_{\tau=t}^{t+W-1}a_i(\tau)\Big]^2.
\end{split}
\]
Then the $W$-slot drift is
\[
\begin{split}
&\Phi(\mathbf{Q}(t+W))-\Phi(\mathbf{Q}(t)) \\
= &\frac{1}{2}\sum_{i} Q_i^2(t+W) - \frac{1}{2}\sum_{i} Q_i^2(t)\\
\le &\frac{1}{2}\sum_{i}\Big[Q_i(t) + \sum_{\tau=t}^{t+W-1} a_i(\tau) - \sum_{\tau=t}^{t+W-1} b_i(\tau)\Big]^2\\
&~~ + \frac{1}{2}\sum_i \Big[\sum_{\tau=t}^{t+W-1} b_i(\tau) +  \sum_{\tau=t}^{t+W-1}a_i(\tau)\Big]^2 - \frac{1}{2}\sum_{i} Q_i^2(t)\\
\le &\sum_{\tau=t}^{t+W-1}  \sum_{i} Q_i(t) \Big( a_i(\tau) -  b_i(\tau)\Big) + 2N B^2 W^2.
\end{split}
\]
Note that for any $\tau\in[t,t+W-1]$ and any $i\in\mathcal{N}$ we have
\begin{equation}\label{eq:q}
Q_i(\tau)-WB\le Q_i(t)\le Q_i(\tau)+WB.
\end{equation}
Then it follows that
\[
\small
\begin{split}
&\sum_{\tau=t}^{t+W-1}  \sum_{i} Q_i(t) \Big( a_i(\tau) -  b_i(\tau)\Big)\\
\le &\sum_{\tau=t}^{t+W-1}  \sum_i  \Big[\big(Q_i(\tau)+WB\big) a_i(\tau)-\big(Q_i(\tau)-WB\big)b_i(\tau)\Big]\\
=& \sum_{\tau=t}^{t+W-1}  \sum_i Q_i(\tau)\Big(a_i(\tau)-b_i(\tau)\Big)+2NB^2W^2,
\end{split}
\]
where the first inequality is due to \eqref{eq:q}.
Therefore, the $W$-slot drift is
\[
\begin{split}
&\Phi(\mathbf{Q}(t+W))-\Phi(\mathbf{Q}(t))\\
\le &\sum_{\tau=t}^{t+W-1} \sum_{i=1}^N Q_i(\tau) \Big(a_i(\tau) - b_i(\tau)\Big) + 4NB^2 W^2.
\end{split}
\]
This completes the proof to Lemma \ref{lm:drift}.
\end{proof}

\vspace{3mm}

Let $\alpha_t^*$, $\mathbf{a^*}(t)$ and $\mathbf{b^*}(t)$ be control action, the controlled arrival vector and the service vector in slot $t$ under the optimal non-causal policy. Then we derive an upper bound on the $W$-slot drift-plus-penalty term.
\begin{lemma}\label{lm:drift-penalty}
The $W$-slot drift-plus-penalty term satisfies
\[
\begin{split}
&\Phi(\mathbf{Q}(t+W))-\Phi(\mathbf{Q}(t)) - \sum_{\tau=t}^{t+W-1}VU(\alpha_t,\omega_t)\\
\le &-\sum_{\tau=t}^{t+W-1} VU(\alpha^*_t,\omega_t)+ 6NB^2 W^2.
\end{split}
\]
\end{lemma}
\begin{proof}
We have
\[
\small
\begin{split}
&\Phi(\mathbf{Q}(t+W))-\Phi(\mathbf{Q}(t)) - \sum_{\tau=t}^{t+W-1}VU(\alpha_t,\omega_t)\\
\le & \sum_{\tau=t}^{t+W-1} \Big[\sum_{i=1}^N Q_i(\tau) \Big(a_i(\tau) - b_i(\tau)\Big) - VU(\alpha_t,\omega_t)\Big]+ 4NB^2 W^2\\
\le & \sum_{\tau=t}^{t+W-1} \Big[\sum_{i=1}^N Q_i(\tau) \Big(a^*_i(\tau) - b^*_i(\tau)\Big) - VU(\alpha^*_t,\omega_t)\Big]+ 4NB^2 W^2,
\end{split}
\]
where the first inequality is due to Lemma \ref{lm:drift} and the second inequality is due to the operation of the Drift-plus-Penalty policy \eqref{eq:dpp}. By \eqref{eq:q}, we have
\[
\small
\begin{split}
&Q_i(\tau) \Big(a^*_i(\tau) - b^*_i(\tau)\Big)\\
\le & (Q_i(t)+WB)a^*_i(\tau)-(Q_i(t)-WB)b^*_i(\tau)\\
\le & Q_i(t) \Big(a^*_i(\tau) - b^*_i(\tau)\Big)+2WB^2.
\end{split}
\]
Plugging the above inequality into the drift-plus-penalty term, we have
\[
\small
\begin{split}
&\Phi(\mathbf{Q}(t+W))-\Phi(\mathbf{Q}(t)) - \sum_{\tau=t}^{t+W-1}VU(\alpha_t,\omega_t)\\
\le & \sum_{i=1}^N Q_i(t) \sum_{\tau=t}^{t+W-1}  \Big(a^*_i(\tau) - b^*_i(\tau)\Big) \\
&-  \sum_{\tau=t}^{t+W-1} VU(\alpha^*_t,\omega_t)+ 6NB^2 W^2\\
\le & -\sum_{\tau=t}^{t+W-1} VU(\alpha^*_t,\omega_t)+ 6NB^2 W^2,
\end{split}
\]
where the last inequality holds because of the window constraints \eqref{eq:AQT}. This completes the proof of Lemma \ref{lm:drift-penalty}.
\end{proof}

\vspace{2mm}

We then divide the time horizon into frames of size $W$ slots. Without loss of generality, assume that $W$ divides $T$ such that the total number of frames is $T\slash W $. Summing the drift-plus-penalty term in Lemma \ref{lm:drift-penalty} over $t=0,W,\cdots,T-W$ and noticing that $\mathbf{Q}(0)=\mathbf{0}$,  we have
\begin{equation}\label{eq:partial}
\begin{split}
&\Phi(\mathbf{Q}(T)-V\sum_{\tau=0}^{T-1}U(\alpha_t,\omega_t)\\
\le &-V\sum_{\tau=0}^{T-1}U(\alpha^*_t,\omega_t)+6NB^2 W^2T\slash W\\
= &-V\sum_{\tau=0}^{T-1}U(\alpha^*_t,\omega_t)+6NB^2 WT.
\end{split}
\end{equation}
Thus, we have
\[
\Phi(\mathbf{Q}(T))\le VT(U_{\max}-U_{\min})+6NB^2 WT,
\]
which implies that
\[
\sum_i Q_i(T)\le \sqrt{N}\sqrt{2\Phi(\mathbf{Q}(t_{K}))}=O\Big(\sqrt{T(V+W)}\Big).
\]
Similarly, by \eqref{eq:partial} we have
\[
\sum_{\tau=0}^{T-1}U(\alpha^*_t,\omega_t)-\sum_{\tau=0}^{T-1}U(\alpha_t,\omega_t)\le 6NB^2 WT\slash V,
\]
which implies that for any sequence of network events $\omega_0,\cdots,\omega_{T-1}$ the utility regret is
\[
\begin{split}
\mathcal{R}_T\Big(\{\omega_0,\cdots,\omega_{T-1}\}\Big)\le &6NB^2 WT\slash V\\
= &O\Big(\frac{WT}{V}\Big).
\end{split}
\]
This completes the proof to Theorem \ref{thm:max-AQT}.

\section{Proof to Theorem \ref{thm:tracking-AQT}}\label{ap:tracking-AQT}
For convenience of notation, let $\alpha_t$, $\mathbf{Q}(t)$, $\mathbf{a}(t)$ and $\mathbf{b}(t)$ be the control action, the queue length vector, the arrival vector and the service vector in slot $t$ under the Tracking Algorithm. Also let $\alpha^*_t$, $\mathbf{Q}^*(t)$, $\mathbf{a}^*(t)$ and $\mathbf{b}^*(t)$ be the control action, the queue length vector, the arrival vector and the service vector in slot $t$ under the optimal solution \eqref{eq:w-opt}. It can be verified that $\{\alpha^*_t\}_{t=0}^{T-1}$ is also an optimal solution to \textbf{NUM}. Time is divided into frames of size $W$ slots.  Without loss of generality, we assume that $W$ divided $T$ and the total number of frames is $R=T\slash W$. 

We introduce two types of ``debt" queues that measure the performance difference between the Tracking Algorithm and the optimal policy. The first type of debt queue is denoted by $q_{i,\omega}(t)$, which measures the ``debt" owned by the Tracking algorithm to the optimal policy w.r.t. the amount of resources allocated to queue $i$ under network event $\omega$ up to time $t$. Its evolution is as follows. 
\begin{itemize}
\item{If $\omega_t\ne \omega$, then 
\[
q_{i,\omega}(t+1) = q_{i,\omega}(t).
\]}
\item{If $\omega_t=\omega$, then 
\[
q_{i,\omega}(t+1) = q_{i,\omega}(t)+\Big(b^*_i(t)-a^*_i(t)\Big)-\Big(b_i(t)-a_i(t)\Big),
\]
where $b^*_i(t)-a^*_i(t)$ is the new debt arrival in slot $t$ and $b_i(t)-a_i(t)$ is the cleared debt in slot $t$. }
\end{itemize}

The second type of debt queue is denote by $u_\omega(t)$, which measures the ``debt" owned by the Tracking algorithm to the optimal policy w.r.t. the gained network utility under network event $\omega$ up to time $t$. Its evolution is as follows.
\begin{itemize}
\item{If $\omega_t\ne \omega$, then 
\[
u_\omega(t+1)=u_{\omega}(t).
\]}

\item{If $\omega_t=\omega$, then 
\[
u_\omega(t+1) = u_\omega(t)+U(\alpha^*_t,\omega)-U(\alpha_t,\omega).
\]}
\end{itemize}

The following lemma gives an upper bound on the length of the two virtual queues.
\begin{lemma}\label{lm:tracking-AQT}
For any $t\in\mathcal{T}$ and any type of network event $\omega\in\Omega$, we have
\[
\begin{split}
q_{i,\omega}(t)&\le WB,~\forall i\in\mathcal{N}\\
u_\omega(t)&\le WU_{\max}.
\end{split}
\]
\end{lemma}
\begin{proof}
Time is divided into frames of size $W$ slots.  In frame 1, suppose that network event $\omega$ occurs $M_1$ times, in slots $t_1,\cdots,t_{M_1}$, respectively. Note that the action queue $\mathcal{Q}_\omega$ is empty throughout frame 1 and no action is taken under the Tracking Algorithm. Thus, the lengths of the two debt queues at the beginning of frame 2 are
\[
\small
\begin{split}
q_{i,\omega}(W)&=\sum_{m=1}^{M_1} \Big(b^*_i(t_m)-a^*_i(t_m)\Big)\le M_1B\le WB,\\
u_{\omega}(W)&=\sum_{m=1}^{M_1} U(\alpha_{t_m}^*,\omega)\le M_1U_{\max}\le WU_{\max}.
\end{split}
\]
Note that the above upper bounds not only apply to time $W$ but also any time within frame 1.

At the beginning of frame 2, there are $M_1$ actions in the action queue $\mathcal{Q}_{\omega}$: $\alpha^*_{t_1},\cdots,\alpha^*_{t_{M_1}}$, and these actions will be taken one by one whenever network event $\omega$ occurs. In frame 2, suppose that network event $\omega$ occurs $M_2$ times. We discuss two scenarios.
\begin{itemize}
\item{If $M_2\ge M_1$, then all the $M_1$ actions in the action queue $\mathcal{Q}_\omega$ are taken, and thus the debts owned in frame 1 are all cleared. Moreover, the new debt owned in frame 2 for the two debt queues is most $M_2B$ and $M_2U_{\max}$. Thus, at the beginning of frame 3, the remaining debt in the two debt queues is upper bounded by $M_2B$ and $M_2U_{\max}$, respectively. }

\item{ If $M_2<M_1$, then the debts owned upon the first $M_2$ occurrences of event $\omega$ in frame 1 are cleared. The new debt owned in frame 2 for the two debt queues is most $M_2B$ and $M_2U_{\max}$. Thus, at the beginning of frame 3, the remaining debt in the two debt queues is at most $M_2B+(M_1-M_2)B=M_1B$ and $M_2U_{\max}+(M_1-M_2)U_{\max}=M_1U_{\max}$, respectively. }
\end{itemize}
Therefore, in both of the above scenarios, we can conclude that at the beginning of frame 3
\[
\small
\begin{split}
q_{i,\omega}(2W)&\le \max\{M_1,M_2\}B\le WB,\\
u_\omega(2W)&\le \max\{M_1,M_2\}U_{\max}\le WU_{\max}.
\end{split}
\]
Note that the above upper bounds not only apply to time $2W$ but also any time within frame 2.

Similar argument applies to any of the subsequent frames $r\ge 3$: for any $t\in [(r-1)W, rW]$
\[
\begin{split}
q_{i,\omega}(t)&\le  B\max_{j\le r}M_j\le WB,\\
u_\omega(t)&\le U_{\max}\max_{j\le r}M_j\le WU_{\max}.
\end{split}
\]
This concludes our proof.
\end{proof}

\vspace{2mm}

For each $i\in\mathcal{N}$, let $\tau_i$ be the last time $t$ when $Q_{i}(t)=0$. Assume that $\tau_i$ is contained in frame $r$ and let $t_i=(r+1)W$ (i.e., the beginning of frame $r+1$). Clearly we have $t_i-\tau_i\le W$ and thus
\begin{equation}\label{eq:qb}
Q_i(t_i) \le Q_i(\tau_i) + \sum_{t=\tau_i}^{t_i-1}a_i(t)\le WB.
\end{equation}
Then it follows that
\[
\begin{split}
Q_i(T)&=Q_i(t_i)+\sum_{t=t_i}^{T-1}\Big(a_i(t)-b_i(t)\Big)\\
&\le WB+\sum_{t=t_i}^{T-1} \Big(a_i(t)+a^*_i(t)-a^*_i(t)-b_i(t)\Big) \\
& \le WB+\sum_{t=t_i}^{T-1} \Big(a_i(t)+b^*_i(t)-a^*_i(t)-b_i(t)\Big)\\
&= WB+\sum_{\omega\in\mathcal{W}}\sum_{t\in\mathcal{T}_{\omega}}\Big[\Big (b^*_i(t)-a^*_i(t)\Big) - \Big(b_i(t)-a_i(t)\Big)\Big].
\end{split}
\]
Here, the first inequality is due to \eqref{eq:qb} and the second inequality is due to Equation \eqref{eq:AQT}. The last equality regroups time slots according to the type of network event that occurred in each slot, where we define
\[
\mathcal{T}_{\omega}=\{t|t_i\le t\le T-1,~\omega_t=\omega\},~\forall \omega\in\mathcal{W}.
\]
Note that by the definition of the debt queue $q_{i,\omega}$ and Lemma \ref{lm:tracking-AQT}, we have
\[
\begin{split}
&\sum_{t\in\mathcal{T}_{\omega}}\Big[\Big (b^*_i(t)-a^*_i(t)\Big) - \Big(b_i(t)-a_i(t)\Big)\Big]\\
\le &q_{i,\omega}(T)-q_{i,\omega}(t_i)\le B W.
\end{split}
\]
Then it follows that
$Q_i(T)\le WB+|\Omega|BW$ for any $i\in\mathcal{N}$ and 
\[
\sum_i Q_i(T)\le NWB+|\Omega|NBW=O(W).
\]
Similarly, for utility regret, we have
\[
\begin{split}
\mathcal{R}_T &= \sum_{t=0}^{T-1}\Big[U(\alpha^*_t,\omega_t)-U(\alpha_t, \omega_t)\Big]\\
&=\sum_{\omega\in\Omega}\sum_{t\in \mathcal{T}'_\omega} \Big[U(\alpha^*_t,\omega)-U(\alpha_t, \omega)\Big]\\
&=\sum_{\omega\in\Omega} \Big[u_\omega(T)-u_\omega(0)\Big]\\
&\le |\Omega|WU_{\max}=O(W),
\end{split}
\]
where we define
\[
\mathcal{T}'_{\omega}=\{t|0\le t\le T-1,~\omega_t=\omega\},~\forall \omega\in\mathcal{W}.
\]
This completes the proof to Theorem \ref{thm:tracking-AQT}.

\section{Proof to Theorem \ref{thm:max-general}}\label{ap:max-general}
We divide the time horizon into frames of size $W$ slots, where the value of $W$ is to be selected later. Following the same line of argument as in Lemma \ref{lm:drift-penalty}, we can derive an upper bound for the $W$-slot drift-plus-penalty term:
\begin{equation}\label{eq:w-drift2}
\begin{split}
&\Phi(\mathbf{Q}(t+W))-\Phi(\mathbf{Q}(t)) - \sum_{\tau=t}^{t+W-1}VU(\alpha_t,\omega_t)\\
\le & \sum_{i=1}^N Q_i(t) \sum_{\tau=t}^{t+W-1}  \Big(a^*_i(\tau) - b^*_i(\tau)\Big) \\
&-  \sum_{\tau=t}^{t+W-1} VU(\alpha^*_t,\omega_t)+ 6NB^2 W^2.
\end{split}
\end{equation}
By the definition of $V_T$-constrained adversary, we have
\[
\sum_{\tau=t}^{t+W-1}  \Big(a^*_i(\tau) - b^*_i(\tau)\Big)\le V_T,~\forall t\in\mathcal{T}, i\in\mathcal{N}.
\]
Plugging the above inequality into \eqref{eq:w-drift2}, we have
\begin{equation}\label{eq:w-drift3}
\begin{split}
&\Phi(\mathbf{Q}(t+W))-\Phi(\mathbf{Q}(t)) - \sum_{\tau=t}^{t+W-1}VU(\alpha_t,\omega_t)\\
\le & V_T\sum_{i=1}^N Q_i(t) -  \sum_{\tau=t}^{t+W-1} VU(\alpha^*_t,\omega_t)+ 6NB^2 W^2.
\end{split}
\end{equation}
Let $T'$ be the time when the total queue length under the Drift-plus-Penalty algorithm reaches the peak, i.e., 
\[
T'\triangleq\arg\max_{t\in\mathcal{T}} \sum_i Q_i(t).
\]
For convenience, we also define $Q_{\max}$ to be the peak queue length under the Drift-plus-Penalty algorithm:
\[
Q_{\max}\triangleq \sum_i Q_i(T') = \max_{t\in\mathcal{T}} \sum_i Q_i(t).
\]
Without loss of generality, we assume that $W$ divides $T'$.  Summing \eqref{eq:w-drift3} over $t=0,W,2W,\cdots,T'-W$, we have
\begin{equation}\label{eq:general-1}
\begin{split}
&\Phi(\mathbf{Q}(T')) - \sum_{\tau=0}^{T'-1}VU(\alpha_t,\omega_t)\\
\le & V_T Q_{\max} \frac{T'}{W} -  \sum_{\tau=0}^{T'-1} VU(\alpha^*_t,\omega_t)+ 6NB^2 W^2\frac{T'}{W}.
\end{split}
\end{equation}
Note that
\[
\Phi(\mathbf{Q}(T')) = \frac{1}{2}\sum_{i} Q^2_i(T')\ge \frac{1}{2N}\Big(\sum_i Q_i(T')\Big)^2=\frac{Q^2_{\max}}{2N}.
\]
\[
\sum_{\tau=0}^{T'-1} \Big[U(\alpha_t,\omega_t)-U(\alpha^*_t,\omega_t)\Big]\le T'(U_{\max}-U_{\min}).
\]
Plugging the above two inequalities into \eqref{eq:general-1} and noticing that $T'\le T$, we have the following quadratic inequality w.r.t. $Q_{\max}$:
\[
\begin{split}
\frac{Q^2_{\max}}{2N}-\frac{V_T T}{W} Q_{\max}\le T(U_{\max}-U_{\min})+6NB^2 WT.
\end{split}
\]
Solving this inequality yields
\begin{equation}\label{eq:solve}
Q_{\max}\le \frac{c_1 V_T T}{W}+\sqrt{\frac{c_2V_T^2 T^2}{W^2}+c_3T(W+V)},
\end{equation}
where $c_1,c_2,c_3$ are constants independent of $W$, $T$ and $V_T$. Minimizing the right-hand side of \eqref{eq:solve} w.r.t. $W$, we have
\[
Q_{\max}=O\Big(\sqrt{V_T^{2\slash 3} T^{4\slash 3}+TV}\Big)=O\Big(V_T^{1\slash 3} T^{2\slash 3}+T^{1\slash 2}V^{1\slash 2}\Big),
\]
where the optimal value of $W$ is $W=c_4 V_T^{2\slash 3}T^{1\slash 3}$ for some constant $c_4>0$. As a result, we can conclude that
\[
\sum_i Q_i(T)\le Q_{\max}\le O\Big(V_T^{1\slash 3} T^{2\slash 3}+T^{1\slash 2}V^{1\slash 2}\Big).
\]

Next we analyze the utility regret achieved by the Drift-plus-Penalty algorithm. Without loss generality, we also assume that $W$ divides $T$. Summing \eqref{eq:w-drift3} over $t=0,W,2W,\cdots,T-W$, we have
\[
\begin{split}
\mathcal{R}_T\Big(\{\omega_0,\cdots,\omega_{T-1}\}\Big)&=\sum_{\tau=0}^{T-1} \Big[U(\alpha_t,\omega_t)-U(\alpha^*_t,\omega_t)\Big]\\
&\le \frac{V_T TQ_{\max}}{WV}+\frac{6NB^2WT}{V}.
\end{split}
\]
Plugging the bound on $Q_{\max}$ and the optimal value of $W$ into the above inequality, we can conclude that
\[
\mathcal{R}_T\Big(\{\omega_0,\cdots,\omega_{T-1}\}\Big)=O\Big(\frac{V_T^{2\slash 3}T^{4\slash 3}}{V}+\frac{V_T^{1\slash 3}T^{7\slash 6}}{V^{1\slash 2}}\Big).
\]

\section{Proof to Theorem \ref{thm:tracking-general}}\label{ap:tracking-general}
Time is divided into frames of size $W$ slots.  Without loss of generality, we assume that $W$ divided $T$ and the total number of frames is $R=T\slash W$.  We first introduce some notations.
\begin{itemize}
\item{Let $\alpha^W_t$, $\mathbf{Q}^W(t)$, $\mathbf{a}^W(t)$ and $\mathbf{b}^W(t)$ be the control action, the queue length vector, the arrival vector and the service vector in slot $t$ under the Tracking Algorithm with parameter $W>0$.}
\item{Let $\alpha^{*W}_t$, $\mathbf{Q}^{*W}(t)$, $\mathbf{a}^{*W}(t)$ and $\mathbf{b}^{*W}(t)$ be the control action, the queue length vector, the arrival vector and the service vector in slot $t$ under the solution to the modified problem \eqref{eq:track-shed} with parameter $W>0$.}
\item{Let $\alpha^*_t$ be the control action in slot $t$ under the optimal policy to \textbf{NUM}.}
\end{itemize}
It can be easily verified that Lemma \ref{lm:tracking-AQT} still holds in $V_T$-constrained networks. Then following the similar line of argument as in the proof to Theorem \ref{thm:tracking-AQT}, we have
\[
\small
\begin{split}
Q^W_i(T)&\le WB+\sum_{t=t_i}^{T-1} \Big(a^W_i(t)+a^W_i(t)-a^W_i(t)-b^W_i(t)\Big) \\
& \le WB+\sum_{t=t_i}^{T-1} \Big(a^W_i(t)+b^{*W}_i(t)-a^{*W}_i(t)-b^W_i(t)\Big)+\frac{V_T T}{W}\\
&= WB+|\Omega|BW+\frac{V_T T}{W},
\end{split}
\]
where the second inequality is due to the first constrain in \eqref{eq:track-shed}. Thus the total queue length is
\[
\sum_i Q_i(T)\le NBW+|\Omega|NBW+\frac{NV_T T}{W}=O\Big(W+\frac{V_T T}{W}\Big).
\]
Similarly, following the analysis in Theorem \ref{thm:tracking-AQT}, we can obtain that 
\[
\begin{split}
\sum_{t=0}^{T-1}\Big[U(\alpha^{*W}_t,\omega_t)-U(\alpha^W_t, \omega_t)\Big]\le |\Omega|WU_{\max}=O(W).
\end{split}
\]
Noticing that the optimal solution $\{\alpha^*\}_{t=0}^{T-1}$ to \textbf{NUM} is also a feasible solution to the modified problem \eqref{eq:track-shed}. Therefore, we have
\[
\sum_{t=0}^{T-1}U(\alpha^{*W}_t,\omega_t) \ge \sum_{t=0}^{T-1}U(\alpha^*_t, \omega_t),
\]
which implies that the utility regret satisfies
\[
\begin{split}
\mathcal{R}_T\Big(\{\omega_0,\cdots,\omega_{T-1}\}\Big) &=\sum_{t=0}^{T-1}\Big[U(\alpha^*_t,\omega_t)-U(\alpha^W_t, \omega_t)\Big]\\
&\le \sum_{t=0}^{T-1}\Big[U(\alpha^{*W}_t,\omega_t)-U(\alpha^W_t, \omega_t)\Big]\\
& = O(W).
\end{split}
\]
Setting the value of $W$ to be $W=\sqrt{TV_T}$ gives the desired result.
\end{document}